\documentclass[12pt,a4wide]{article}

\usepackage[authoryear]{natbib}
\usepackage{graphicx,color,enumerate,amssymb}
\usepackage{amsmath}
\usepackage{amssymb}
\usepackage{enumerate}
\usepackage{subcaption}
\usepackage{amsthm}
\usepackage{xcolor}
\usepackage{mathabx}
\usepackage{array,multirow}

\newcommand*\rot{\rotatebox{90}}

\makeatletter
\def\R{{\rm I\kern-.2em R}}
\makeatother

\newtheorem{thm}{Theorem}[section]
\newtheorem{lemma}[thm]{Lemma}
\newtheorem{remark}[thm]{Remark}

\newtheorem{assump}{Assumption}

\newtheorem{assumpT}{Assumption}

\newcommand{\eps}{\varepsilon}

\newcommand{\bb}[1]{\boldsymbol{#1}}

\newcommand{\wtlambda}{\widetilde{\lambda}}

\newcommand{\E}{\mathbb{E}\,}
\newcommand{\tth}{\boldsymbol{\theta}}
\newcommand{\tTheta}{\boldsymbol{\Theta}}
\newcommand{\tX}{\bb{X}}
\newcommand{\tZ}{\bb{Z}}
\newcommand{\tx}{\bb{x}}

\newcommand{\tlam}{\bb \lambda}

\newcommand{\tgamma}{\bb \gamma}

\newcommand{\Pto}{\stackrel{\mathsf{P}}{\rightarrow}}
\newcommand{\asto}{\stackrel{\mathsf{a.s.}}{\rightarrow}}
\newcommand{\Dto}{\stackrel{\mathsf{D}}{\rightarrow}}

\usepackage{setspace} 

\begin{document}
%%%%%%%%%%%%%%%%%%%%%%%%%%%%%%%%%%%%%%%%%%%%%%%%%%%%%%%%%%%%%%%%%%%%%%%%%%%
\begin{center}
{ \LARGE\sc Specifications tests for count time series models with covariates \\
} 
\vspace*{1cm} 

{ \large\sc \v{S}\'arka Hudecov\'a$^{\rm{a}}$, \large\sc Marie Hu\v{s}kov\'a$^{\rm{b}}$, \large\sc Simos G. Meintanis}$^{\rm{c,d}}$, \\ \vspace*{0.5cm}
%{\it{$^{\rm{2}}$Unit for Business Mathematics and Informatics, North--West University, Potchefstroom, South Africa}}
\end{center}

{\small 
\noindent$^{\rm{a}}$ Department of Probability and Mathematical Statistics, Charles University, Prague, Czech Republic, email: hudecova@karlin.mff.cuni.cz\\[2ex]
$^{\rm{b}}$   Department of Probability and Mathematical Statistics, Charles University, Prague, Czech Republic, email: huskova@karlin.mff.cuni.cz \\[2ex]
$^{\rm{c}}$  Department of Economics, National and Kapodistrian University
of Athens, Athens, Greece, email: simosmei@econ.uoa.gr \\
$^{\rm{d}}$ Pure and Applied Analytics, North--West University, Potchefstroom, South Africa \\

}
%\today

%
%
%\newpage
%
%
%\noindent{\large \sc Running head: Tests for count time series}
%
%\vskip5mm
%
%\noindent{\bf Corresponding author:} Simos G. Meintanis, email: simosmei@econ.uoa.gr 
%
%

\vskip5mm

\noindent {\small {\bf Abstract.} We propose a goodness--of--fit test for a class of  count time series models with covariates which includes  the Poisson autoregressive model with covariates (PARX) as a special case. The test criteria are derived from a specific characterization for the conditional probability generating function and the test statistic is formulated as a $L_2$ weighting norm of the corresponding sample counterpart. The asymptotic properties of the proposed test statistic are provided under the null hypothesis as well as under specific alternatives. A bootstrap version of the test is explored in a Monte--Carlo study and illustrated on a real data set on road safety.} \vspace*{0.3cm}


\newpage

\section{Introduction}
Let $\{Y_t\}$ be a time series of counts, i.e. a time series of variables taking values in $\mathbb{N}_0$, and let  $\{\mathbf{X}_t\}$ be a $m$-dimensional vector time series of exogenous covariates. Denote as ${\cal{I}}_t$ the information set available up to time $t$ and   let
$\{\bb Z_t\}$ be a time series of $d$-dimensional random vectors such that $\bb Z_t$  may contain past values of both endogenous as well as exogenous covariates.
Our aim is to develop procedures for parametric specifications of the conditional distribution of $Y_t$ given $\mathcal{I}_{t-1}$, i.e. we are interested in testing the null hypothesis
\begin{equation}\label{null}
	{\cal{H}}_0: Y_t|{\cal {I}}_{t-1} \sim {\cal {F}}_{\lambda_t}, \quad \lambda_t = h(\bb Z_t;\tth),
\end{equation} 
where ${\cal{F}}_{\lambda}$ denotes a fixed  count distribution with a mean $\lambda$, 
 $h$ is a specified function and  $\tth$ is an unspecified vector of parameters.
Notice that the manner in which $\lambda_t$ depends on the covariate vector $\bb Z_t$ is  fixed under ${\cal{H}}_0$ apart from a finite-dimensional parameter $\tth$ which must be estimated from the data.

We will restrict ourselves to a class of distributions $\{\mathcal{F}_{\lambda}\}$ such that  %indexed by the mean
%\cite{francq21} considered a general class of models for distributions ${\cal {F}}_{\lambda}$ satisfying a specific stochastic order condition. Namely, it is assumed that 
$\lambda>0$ is the mean of ${\cal {F}}_{\lambda}$ and  the family $\{\mathcal{F}_{\lambda}\} $ satisfies 
\begin{equation}\label{eq:SOC}
\lambda_1\leq \lambda_2 \Rightarrow {\cal {F}}_{\lambda_1}(y)\geq {\cal {F}}_{\lambda_2}(y) \quad \text{ for all } y\in\R.
\end{equation}
%\textcolor{blue}{\it It suffices to take $y\in\mathbb{N}_0$ above.}
The condition \eqref{eq:SOC}  holds, for instance, for Poisson distribution with mean $\lambda$, for Negative Binomial distribution with mean $\lambda$ and fixed $\phi$, or for a zero--inflated Poisson distribution. 
\cite{francq21} provide conditions for existence of a stationary and ergodic solution to \eqref{null} with a conditional distribution satisfying \eqref{eq:SOC}. 
%We will consider this class of distributions and construct a goodness-of-fit test for the null hypothesis ${\cal{H}}_0$. 

 If  ${\cal{F}}_\lambda$ denotes the Poisson distribution with mean $\lambda$ and if no exogenous variables are considered
%, so $\tZ_t$ consists solely of  lagged  values of  $Y$ and lagged values of $\lambda$, 
then  \eqref{null} corresponds to  the INGARCH model  of \citep{FRT09}. %(For a recent generalization incorporating also time--varying dispersion we refer to \cite{BPFO}).  
An extension to this model, abbreviated sometimes as PARX, incorporating also exogenous variables was proposed and studied by  \cite{agosto16}, and
% proposed the so--called PARX model, which  also incorporates exogenous variables within the autoregression function. 
the  model has already found a number of important extensions and applications in various fields. 
In finance, \cite{agosto16} used PARX to model US corporate default counts, while
\cite{agosto2019} analyzed Italian corporate default counts data in the real estate sector. 
%,cornea},
Other applications are in sports \citep{angelini17}, and disease modelling \citep{agosto2020}.
Alternative specifications for $\mathcal{F}_{\lambda}$ include the Negative Binomial (NB) distribution %with mean $\lambda$ and variance $\lambda(1+\lambda/\phi)$, with $\phi>0$ regulating overdispersion 
\citep{christou,zhu,fokianos17}, and the zero--inflated Poisson distribution \citep{francq21}, among others.

The problem of testing goodness-of-fit within INGARCH models has been considered by many authors, for instance \citep{fokianos2013,weiss+schweer2015, schweer2016, HHM-gof, HHM-mult} to mention some of them.   We propose a method for testing $\mathcal{H}_0$ for an arbitrary family of conditional distributions satisfying \eqref{eq:SOC}, a possibly non-linear function $h$ and a setup with external covariates. This, of course, involves %testing the null hypothesis ${\cal{H}}_0$ of the Poisson specification against alternative conditional law specifications, in presence of the exogenous covariates in the model, so the procedure is suitable for 
testing  PARX and INGARCH specifications as special cases. Our test is based on a characterization of the conditional mean by \cite{bierens}, which has ben successfully used in construction of various specification tests,  e.g.,  see \citep{escanciano,hlavka2017} and references therein. In contrast to these works, we reformulate the characterization 
 in terms of the conditional probability generating function (PGF) of $Y_t$ given $\bb Z_t$.
PGF--based procedures are specifically tailored for count responses and are often  convenient from the computational point of view. Earlier PGF-based methods, suggested by \cite{HHM-gof}, \cite{schweer2016} and \cite{HHM-mult}, are suitable for an analysis of time series without exogenous covariates. In contrast to this, the new tests suggested herein, although still PGF--based,  adopt an alternative formulation which allows them to be readily applied  to more complicated dynamic models.

 The rest of the paper is outlined as follows. Section~\ref{sec:model} introduces the considered model. The test statistic is proposed in Section~\ref{sec:gen}, while its asymptotic properties are studied in Section \ref{sec:as}.
 Section \ref{sec:boot} describes the practical computation of the test and presents a bootstrap version. 
 %procedure which is suitable in order to approximate its null distribution. 
  An application to the 
Poisson autoregression with covariates (PARX) is discussed in Section \ref{sec:PARX}.  Section \ref{sec:simul} presents  results of a Monte Carlo study. A real data application of road safety is included in Section  \ref{sec:data}. We finally conclude with discussion of our finding and further outlook in  Section \ref{sec:concl}. Proofs are  postponed to Appendix. %Section~\ref{sec:proofs}.

\section{Model specification}\label{sec:model}

Let $p\geq 1$ and $q\geq 0$ be given integers.
Let $\{\mathcal{F}_{\lambda}, \lambda \in \Omega \subset (0,\infty)\}$ be some family of discrete distributions on $\mathbb{N}_0$ indexed by a mean parameter   $\lambda\in(0,\infty)$ and satisfying~\eqref{eq:SOC}.  Consider  a   model
\begin{equation}\label{rate2}
Y_t|{\cal {I}}_{t-1} \sim \mathcal{F}_{\lambda_t}, \quad  \lambda_t=h(\bb Z_t;\tth),
\end{equation}
where $\tth\in \boldsymbol{\Theta}\subset \R^K$ is a vector of unknown model parameters,
 \begin{equation}\label{eq:Zt}
  \bb Z_t = (Y_{t-1},\dots,Y_{t-p},\lambda_{t-1},\dots,\lambda_{t-q}, \bb X_{t-1}^\top)^\top,
  \end{equation} 
 is a vector of dimension $d=p+q+m$,  which contains past values of both endogenous as well as exogenous covariates, and
  $h:[0,\infty)^{p+q}\times \R^{m+K}\to(0,\infty)$ is a  link function such that
 %$h:\R^{p+q+m+K}\to(0,\infty)$ is a  link function such that
\begin{equation}\label{eq:hNONLIN}
h(\bb z;\tth)=h(\bb y, \bb l,\bb x; \tth)= r(\bb y,\bb l;\tth)+\pi(\bb x;\tth) 
\end{equation}
for $\bb z\in[0,\infty)^{p+q}\times\R^{m}$ such that $\bb z=(\bb y^\top,\bb l^\top,\bb x^\top)^\top$, $ \bb y\in[0,\infty)^p, \bb l\in[0,\infty)^q, \bb x \in\R^m$ and for some non-negative functions $r:[0,\infty)^{p+q}\times \R^K\to (0,\infty)$ and $\pi:\R^{m+K}\to(0,\infty)$. This means that the conditional mean of $Y_t$ given $\mathcal{I}_{t-1}$ is specified as
\begin{equation*} %\label{eq:NONLIN}
\E[Y_t|\mathcal{I}_{t-1}]= \lambda_t = r(Y_{t-1},\dots,Y_{t-p},\lambda_{t-1},\dots,\lambda_{t-q};\tth)+\pi(\bb X_{t-1};\tth).
 \end{equation*}
Notice that the random vector $\bb Z_t$ implicitly depends on the value of the model parameter $\tth$, but we do not stress this dependence if not necessary.  If $q=0$ then the vector $\bb Z_t$ contains  only past values of $Y_t$ and $\bb X_t$,  and the conditional mean is specified using a (nonlinear) autoregressive model with exogenous covariates (ARX).

If the exogenous covariate process $\{\bb X_t\}$ is stationary and ergodic, %with non-negative components, 
then  there exists a stationary and ergodic  solution $\{Y_t\}$ to the model \eqref{rate2} with the link function $h$ given by \eqref{eq:hNONLIN}  provided that 
the function $r$ satisfies  
\begin{equation}\label{eq:Lipsch}
%|r(y_1,\dots,y_q,\lambda_1,\dots,\lambda_p,\tth) - r(\widetilde{y}_1,\dots,\widetilde{y}_q,\widetilde{\lambda}_1,\dots,\widetilde{\lambda}_p,\tth)|\leq \sum_{i=1}^q \alpha_i|y_i-\widetilde{y}_i| + \sum_{j=1}^p \beta_j |\lambda_j - \widetilde{\lambda}_j|
|r(\bb y_1,\bb \lambda_1;\tth) - r(\bb y_2, \bb \lambda_2;\tth)|\leq \sum_{i=1}^p \alpha_i|y_{1i}-{y}_{2i}| + \sum_{j=1}^q \beta_j |\lambda_{1j} - {\lambda}_{2j}|
\end{equation}
for all $\bb y_i= (y_{1i},\dots,y_{1p})^\top\in[0,\infty)^p$, $\bb \lambda_i = (\lambda_{i1},\dots,\lambda_{iq})^\top\in[0,\infty)^q$, $i=1,2$,
and 
for some $\alpha_i,\beta_j$, $i=1,\dots,p$, $j=1,\dots,q$  such that  
\begin{equation}\label{eq:alphabeta}
 \sum_{i=1}^p \alpha_i + \sum_{j=1}^q \beta_j <1,
\end{equation}
see \citep{francq21}.

Model \eqref{eq:hNONLIN} includes as a special case a linear model,  which can be for non-negative components of $\tX_t$ formulated as \eqref{rate2} for a link function
  \begin{equation}\label{eq:hLIN}
h(\bb y, \bb l,\bb x; \tth)= \omega+\bb \alpha^\top\bb y + \bb \beta^\top \bb l + \bb \pi^\top \bb x
  \end{equation}
%   $\tth=c(\omega,\alpha_1,\dots,\alpha_p, \beta_1,\dots,\beta_q,\bb \pi^\top)$ and
 for $\tth=(\omega,\bb \alpha^\top,\bb \beta^\top,\bb\pi^\top)^\top$ with $\omega>0$, $\bb\alpha\in[0,1]^p$, $\bb \beta\in[0,1]^q$ and $\bb \pi \in (0,\infty)^m$. The condition for existence of a stationary and ergodic solution then reduces to \eqref{eq:alphabeta}. In that case
  \begin{equation}\label{eq:LIN}
\E[Y_t|\mathcal{I}_{t-1}]= \lambda_t = \omega+\sum_{i=1}^p\alpha_i Y_{t-i} + \sum_{j=1}^ q \beta_j\lambda_{t-j} + \bb \pi^\top \tX_{t-1}.
\end{equation}

\section{Goodness-of-fit test}\label{sec:gen}

 Let $\{(Y_t,\bb X_{t-1})\}_{t=1}^T$ be  observed data. We develop test for the null hypothesis
\begin{align*} %\label{eq:H0}
\mathcal{H}_0: & \{(Y_t,\bb X_{t-1})\}_{t=1}^T \text{ follow model } \eqref{rate2} \text{ with }
\text{a specified class of distributions } \mathcal{F}_{\lambda}, \notag \\ 
& \text{ and
specified function }  h(\cdot;\bb \theta_0)
   \text{ for some  inner point } \bb \theta_0 \in \bb \Theta_0,
\end{align*}
against a general alternative  ${\cal{H}}_1: \text{not} \ 	{\cal{H}}_0$, where $\bb \Theta_0$ stands for a compact subset of the parametric space $\bb \Theta$. 
 The model orders $(p,q)$ are implicitly specified by the function $h$, so they are assumed to be known under the null hypothesis.

Our test utilizes the characterization of \cite{bierens} which for a random variable $U$ and a $d$-dimensional random vector $\bb W$ may be formulated  via the following equivalence relation
\begin{equation}\label{Bierens}
\E[U|{\bb{W}}]=\mu_0({\bb{W}}) \Longleftrightarrow \E[\{U-\mu_0({\bb{W}})\}e^{{\texttt{i}} \bb v^\top \bb W}]=0,  \ \forall \bb v\in \R^d,
\end{equation}
where ${\texttt{i}}=\sqrt{-1}$ and $\mu_0(\cdot)$ denotes a specific regression function possibly involving an unknown parameter vector.

 In the current context, we  rephrase  (\ref{Bierens}) for model \eqref{rate2}  in terms of the conditional  probability generating function (PGF). 
 Recall that a PGF of a count variable $Y$ is defined as $g_Y(u):=\mathbb E(u^Y)$ for $u\in\R$ for which the expectation is finite. 
The PGF is always well defined  for $u\in[0,1]$ and 
the distribution of $Y$ is uniquely determined from the values of $g_Y$ from a neighborhood of the origin. 
 Let $g_{\lambda}(u)$ be the PGF corresponding to ${\cal F}_{\lambda}$.
Under $\mathcal{H}_0$ it holds that %for fixed $u\in (0,1)$ that
\[
\E[u^{Y_t}|{\cal {I}}_{t-1}]=\E[u^{Y_t}| \bb Z_t]=g_{\lambda_t}(u) =g_{h(\bb Z_t;\tth_0)}(u)  , \quad u\in[0,1].
\]
%for $ \lambda_t= h(\bb Z_t;\tth_0)$. 
Hence, 
the characterization featuring in (\ref{Bierens}) can be under ${\cal{H}}_0$ reformulated %for model \eqref{rate2} in terms of the PGF 
as
\begin{equation} \label{null1}
\E[\{u^{Y_t}-g_{\lambda_t}(u)\}e^{{\texttt{i}} \bb v^\top \bb Z_t}]=0, \forall u\in(0,1), \bb v\in \R^d, \quad \lambda_t= h(\bb Z_t;\tth_0).
\end{equation}
 %Following the approach of  \cite{hlavka2017} adapted to the current context,  
 Our test procedure will be based on a weighted $L_2$ norm of an estimator of the expectation  in~\eqref{null1}.
Let  $\widehat {\tth}:=\widehat {\tth}_T$ be a suitable estimator of the model parameter $\tth$ and define 
 \begin{equation} \label{eq:EPS}
	\widehat \varepsilon_t(u)=u^{Y_{t}}- g_{\widehat{\lambda}_t}(u), \ t=1,...,T, \end{equation}
where  $\widehat{\lambda}_{t}$ are for $t=1,\dots,T$,  defined recursively as
$
\widehat {\lambda}_t= h(\widehat{\bb Z}_t;\widehat{\tth}), 
$
with  
$\widehat{\bb Z}_t $ being an estimated version of $\bb Z_t$, computed from some initial values for $Y_t$,  $\widehat{\lambda_t}$, and $\bb X_t$ for $t\leq 0$, i.e. 
%Namely, for 
$
 \widehat{\bb Z}_t=({Y}_{t-1},\dots,{Y
 }_{t-p},  \widehat{\lambda}_{t-1},\dots,\widehat{\lambda}_{t-q},\boldsymbol X^{\top}_{t-1})^\top
$ for $t\geq p+1$. 
 Let $W:\R^{d+1}\to (0,\infty)$ be a nonnegative weight function. Define the test statistic
\begin{equation} \label{teststCM-gen}
\Delta_{T,W}= T \int_{0}^1 \int_{\R^d} \left| D_T(u;\bb v) \right|^2
W(u,\bb v) {\rm{d}}\bb v {\rm{d}}u,
\end{equation}
where
\begin{equation*} % \label{Dn}
D_T(u;\bb v)=\frac{1}{T} \sum_{t=1}^T {\widehat {\varepsilon}}_t(u)e^{{\texttt{i}} \bb v^\top \widehat{\bb Z}_t}, \ (u,\bb v)\in (0,1) \times \R^d.
 \end{equation*}
The test statistic $\Delta_{T,W}$   in (\ref{teststCM-gen}) is an $L_2$--type distance statistic which, at least for large sample size $T$, should be small under the null hypothesis ${\cal {H}}_0$ and large under alternatives. Therefore large values of the test statistic indicate that the null hypothesis is violated.

The estimator $\widehat{\tth}$ can be either a maximum likelihood estimator (MLE), or, more generally, a Poisson quasimaximum likelihood estimator (QMLE). The latter method has been studied in \cite{ahmad} and \cite{francq21}, where it is shown that the resulting estimator is consistent and asymptotically normal under mild regularity conditions.

For practical computation of $\Delta_{T,W}$, it is useful to notice that 
\begin{equation} \label{Dnsum}
|D_T(u;\bb v)|^2=\frac{1}{T^2} \sum_{t,s=1}^T  {\widehat {\varepsilon}}_t(u){\widehat {\varepsilon}}_s(u)\cos\left(\bb v^\top (\bb {\widehat {Z}}_t-\bb {\widehat {Z}}_s)\right).
\end{equation}
If the weight function can be decomposed as $W(u,\bb v)=w(u)\omega(\bb v)$, with $w(\cdot)$ and $\omega(\cdot)$ satisfying some additional assumptions, then further simplification is possible; see Section~\ref{sec:boot} %e refer to  Section~\ref{sec:PARX} 
for more details on the practical aspects of computing the test statistic $\Delta_{T,W}$ and Section~\ref{sec:PARX} for the special case of a linear PARX model. % Poisson $F_{\lambda}$. 

\begin{remark}
Note that the integral  in \eqref{teststCM-gen} is computed for $u\in[0,1]$, which is an approach consistent with  a number of previous works. Alternatively, the integration can be carried over the interval $[-1,1]$, or over an interval $I\subset[-1,1]$ containing the origin.
\end{remark}

\begin{remark}
Our vector of covariates $\bb Z_t$ defined as in~\eqref{eq:Zt} may also include an intercept. In such case, $\widecheck{\bb Z}_t=(1,Y_{t-1},\dots,Y_{t-q},\lambda_{t-1},\dots,\lambda_{t-q},\tX_{t-1}^\top)^\top =(1,\bb Z_t^\top)^\top$, where $\tZ_t$ is given \eqref{eq:Zt} and $d=p+q+m$, then it follows from \eqref{Dnsum} and \eqref{teststCM-gen} that the resulting test statistic $\widecheck{\Delta}_{T,\widecheck{W}}$ based on a weight function $\widecheck{W}:\R^{d+2}\to (0,\infty)$ is equal to  $\Delta_{T,W}$ computed from $\tZ_t$ and the weight function $W(u,v_1,\dots,v_d)=\int_{-\infty}^{\infty} \widecheck{W}(u,x,v_1,\dots,v_d) \mathrm{d} x$. 
\end{remark}

\section{Asymptotic distribution}\label{sec:as}

 This section describes the asymptotic behavior of the test statistic   $\Delta_{T,W}$ under the null hypothesis (Section~\ref{sec:asH0}) as well as under some alternatives (Section~\ref{sec:asALT}). Since some of the assumptions can be substantially weakened for an autoregressive type of model, where $q=0$, the behavior under the null for these models  is treated in more detail in Section~\ref{sec:asAR}.

\subsection{Behavior under the null hypothesis}\label{sec:asH0}

We start with some assumptions on the link function and the weight function.

\begin{assump}\label{as:h}
 %Let $p\geq 1$ and $q\geq 0$ be given integers and
 Let  $h:[0,\infty)^{p+q}\times \R^{m+K}\to(c_h,\infty)$ for some $c_h>0$ be a function of the form \eqref{eq:hNONLIN}
%\[
%h(\bb z;\tth)=h(\bb y, \bb l,\bb x; \tth)= r(\bb y,\bb l;\tth)+\pi(\bb x;\tth) %\quad \bb y\in\R^p, \bb l\in\R^q, \bb x \in\R^m, \tth\in \R^K, \bb z\in\R^{p+q+m}
%\]
%for $\tth \in \R^K$,  $\bb z\in\R^{p+q+m}$ such that $\bb z=(\bb y^\top,\bb l^\top,\bb x^\top)^\top$, $ \bb y\in\R^p, \bb l\in\R^q, \bb x \in\R^m$, and 
for some continuous non-negative functions $r:[0,\infty)^{p+q}\times \R^K\to (0,\infty)$  and $\pi:\R^{m+K}\to(0,\infty)$.
Let $\bb \Theta\subset\R^K$  be such that
 \eqref{eq:Lipsch} holds for all $\tth\in\bb \Theta$  and for some $\alpha_i,\beta_j$, $i=1,\dots,p$, $j=1,\dots,q$, satisfying \eqref{eq:alphabeta}.
\end{assump}

\begin{assump}\label{as:stat}
Let $\{\tX_t\}$ be a strictly stationary and ergodic sequence of $m$-di\-men\-sional random vectors.

Let $\{\mathcal{F}_\lambda\}$ satisfy \eqref{eq:SOC} and let
 $\{Y_t\}$ 
 be the stationary and ergodic solution of~\eqref{rate2} with 
 $\lambda_t=h(\bb Z_t; \tth_0)$, where $\bb Z_t$ is defined by \eqref{eq:Zt}, for a true parameter $\tth_0\in\mathsf{int}(\tTheta_0)$, where $\tTheta_0$ is a compact subset of $\tTheta\subset\R^K$. Assume further that $\E \|\tZ_t\|^2<\infty$. 
\end{assump}
(Here $\mathsf{int}(A)$ stands for the interior of a set $A$ and  $\|\tx\|$ is the Euclidean norm of a vector $\tx$).

The existence of a stationary ergodic solution from Assumption~\ref{as:stat} is proved for $h$ satisfying Assumption~\ref{as:h}  in \cite[Theorem 3.3]{francq21}. Under the same conditions, the process $(Y_t,\bb X_t^\top)^\top$ is stationary and ergodic and the same holds for the multivariate process $\{\bb Z_t\}$. Sufficient conditions for the existence of moments for the linear case with $p=q=1$ are provided in \cite[Theorem 3.2]{francq21}. Also notice that we assume $\tth_0$ lies in the interior of the parameter space $\tTheta_0$; for estimation and specification testing which allow the parameter to lie on the boundary of the parameter space we refer to  \cite{FZ07} and \cite{cav2023}.

\begin{assump}\label{as:r}
Let  $r$ and $\pi$ be the functions from Assumption~\ref{as:h}.  Let $r(\bb y,\cdot,\cdot)$ be three times continuously differentiable on $(0,\infty)^q\times \tTheta$ for all $\bb y\in(0,\infty)^p$ and 
denote  
\[
r_{\bb \lambda}(\bb y,\bb \lambda;\tth) = \frac{\partial r(\bb y,\bb \lambda;\tth)}{\partial \bb \lambda}, \quad r_{\tth}(\bb y,\bb \lambda;\tth) = \frac{\partial r(\bb y,\bb \lambda;\tth)}{\partial \tth} % \quad r_{\bb \lambda \tth} (\bb y,\bb l;\tth)=  \frac{\partial^2 r(\bb y,\bb l;\tth)}{\partial \bb l \partial \bb \tth^\top}, 
\]
for $\bb \lambda\in\R^q$ and $\tth\in\bb \Theta$.
Suppose $r_{\bb \lambda}(\bb y,\bb \lambda;\tth)=d(\tth)$ is constant with respect to $\bb y$ and $\bb \lambda$
%and assume that all these derivatives exist finite for all $\quad \bb y\in\R^p, \bb l\in\R^q$ and $\tth\in\Theta$. 
%Let $r_{\tlam}(\bb Y_{t-1:t-p},\bb \lambda_{t-1:t-q};\tth)$ be 
such that
%a non-random vector such that the matrix
\[
\sup_{\tth \in \Theta_0}  \rho (\bb D) <1, \quad \bb D = \begin{pmatrix} d(\tth) 
\\ \bb I_{q-1}\quad \bb 0\end{pmatrix}, 
\]
where $\rho(\cdot)$ is a spectral radius and $\bb I_{k}$ is a $k\times k$ identity matrix. 
Let $\pi$ be two times continuously differentiable with respect to $\tth$ and denote
\[
\quad \pi_{\tth}(\bb x;\tth) = \frac{\partial \pi(\bb x;\tth)}{\partial \tth}.
\]
Assume that for the stationary variables $Y_t,\lambda_t,\tX_t$  there exists $\kappa>0$ such that 
\begin{equation}\label{eq:Esup}
\E\sup_{\tth \in \Theta_0}  \|\pi_{\tth}(\bb X_{t-1};\tth) \|^\kappa<\infty, \quad  \E \sup_{\tth \in \Theta_0} \|r_{\tth}(\bb Y_{t-1:t-p},\bb \Lambda_{t-1:t-q};\tth) \|^\kappa<\infty,
\end{equation}
where $\bb Y_{t-1:t-p} = (Y_{t-1},\dots,Y_{t-p})^\top$ and $\bb \Lambda_{t-1:t-q} = (\lambda_{t-1},\dots,\lambda_{t-q})^\top$.
\end{assump}

Assumption~\ref{as:r} is related to the smoothness of the mean function and the required conditions ensure that the effect of the initial values  is asymptotically negligible. A similar set of assumptions is invoked by \cite{francq21} to prove  the asymptotic normality of the QMLE. Notice that Assumption~\ref{as:r}  is always satisfied by a linear function \eqref{eq:hLIN}.

\begin{assump}\label{as:g}
Assume that the PGF $g_{\lambda}(u)$ of the distribution $\mathcal{F}_{\lambda}$,  considered  as a function  $g_z(u)$ of $(u,z)$, is twice continuously differentiable and that
 %that 
 % $\partial^2 g_z(u)/\partial z^2$ exist continuous  for all $u\in[0,1]$ and $z\in(0,\infty)$ 
% and
 \[
\left| \frac{\partial g_z(u)}{\partial z}\right|\leq M_1,\quad  \left| \frac{\partial^2 g_z(u)}{\partial z^2}\right|\leq M_2, \quad \left| \frac{\partial^2 g_z(u)}{\partial z \partial u}\right|\leq H(z)
 \]
for all $(u,z)^\top\in(0,1)\times(0,\infty)$, and some  constants $M_1,M_2<\infty$, where $H(\cdot)$ denotes a finite positive function. %for all $u\in(0,1)$ and all  $z\in(0,\infty)$
 \end{assump}

Assumption~\ref{as:g} ensures that the PGF $g_{\lambda}$ as a function $\lambda$ is sufficiently smooth.  For a Poisson distribution, one can take $M_1=M_2=1$ and  $H(z)=z+1$.

\begin{assump}\label{as:lambda}
Let $(Y_t,\lambda_t(\tth))$ be the stationary solution of \eqref{rate2} for parameter $\tth\in\tTheta$. % and let $\widetilde{\lambda}_t(\tth)$ be computed recursively from the same model and parameter $\tth$ for some initial values of $\widetilde{\lambda}_t(\tth)$, $t\leq 0$. 
%Assume that
%\[
%\E \left\| \frac{\partial \lambda(\tth)}{\partial \tth}\Big\vert_{\tth=\tth_0}\right\|<\infty. 
%\]
%Furthermore, 
Let the second derivative
$
\frac{\partial^2 \lambda_t(\tth)}{\partial \tth \partial \tth^\top} %\quad \frac{\partial^2 \widetilde{\lambda}(\tth)}{\partial \tth \partial \tth^\top}
$
exist and be continuous in $\tth$ and let  there exist a neighborhood $\mathcal{V}(
\tth_0)$ of $\tth_0$ such that 
\begin{equation}\label{as:Elam1}
\E \sup_{\tth \in \mathcal{V}(\tth_0)} \frac{\partial \lambda_t(\tth)}{\partial \tth} \left( \frac{\partial \lambda_t(\tth)}{\partial \tth} \right)^\top  \quad \text{ and } \quad \E \sup_{\tth \in \mathcal{V}(\tth_0)} \frac{\partial^2 \lambda_t(\tth)}{\partial \tth\partial \tth^\top}  
\end{equation}
are finite. Furthermore, assume that 
\begin{equation*} %\label{eq:finiteH}
\E H^2(\lambda_t(\tth_0))<\infty, 
\end{equation*}
where the function $H(\cdot)$ is from Assumption~\ref{as:g}.  
\end{assump}

\begin{assump}\label{as:est}
Let $\widehat{\tth}_T$ be an estimator of $\tth$   satisfying 
\[
\sqrt{T}(\widehat{\tth}_T-\tth_0)=\frac{1}{\sqrt{T}}\sum_{i=1}^T \bb s_t (\tth_0, Y_t, \mathcal{I}_{t-1})+o_{\mathsf P}(1),
\]
where  $\tth_0 \in \bb\tTheta_0$ is the true value of $\tth$,  $\{\bb s_t(\tth_0, Y_t, \mathcal{I}_{t-1})\}$ is  %stationary and
a  stationary martingale difference sequence
with respect to the filtration $\{\mathcal{I}_t\}$
with a finite variance such that
$T^{-1/2}\sum_{t=1}^T \bb{s} _t(\tth_0, Y_t, \mathcal{I}_{t-1})\stackrel{D}{\to} \mathsf{N}(\bb{0},\bb{\Sigma})$ for some $K\times K$  positive semidefinite matrix $\bb{\Sigma}$.
\end{assump}

Assumption \ref{as:est} is  satisfied under some additional regularity conditions by the MLE, see \cite{agosto16} for a linear Poisson model. More generally, the estimator  
$\widehat{\tth}_T$ may be the Poisson QMLE, see \cite{ahmad,francq21}.

\begin{assump}\label{as:W}
Let %the weight function 
$W:\R^{d+1}\to(0,1)$ be  $W(u,\bb v)=w(u)\omega(\bb v)$ for some measurable %non-negative 
functions $w:(0,1)\to (0,\infty)$ and $\omega:\R^d\to(0,\infty)$ 
such that $\omega(\bb v)= \omega(-\bb v)$ for all %$u\in[0,1], 
$\bb v \in\R^d$  and %it holds $w(u)>0$, $\omega(\bb v)>0$,
\[
%\omega(\bb v)= \omega(-\bb v), \quad 
0<\int_{0}^1 w(u) \mathrm{d} u<\infty, \quad 0<\int_{\R^d}|| \bb v||^k \omega(\bb v) \mathrm{d}\bb v<\infty
\]
for $k\leq 4$. 
%and
%\[
%0<\int_{0}^1 w(u) \mathrm{d} u<\infty \text{ and } 0<\int_{\R^d}|| \bb v||^2 \omega(\bb v) \mathrm{d}\bb v<\infty.
%\]
\end{assump}

 \smallskip   

Define for $\bb v_1, \bb v_2\in\mathbb R^d$ and $u\in[0,1]$
\begin{align*}
\alpha(\bb v_1,\bb v_2)&=\cos\left(\bb v_1^\top \bb v_2\right)+\sin\left(\bb v_1^\top \bb v_2 \right),\\
%\end{align}
\bb\beta(u,\bb v_1)&=-\E \left\{
\frac{\partial g_z(u)}{\partial z}\Bigr\vert_{z=\lambda_1} \frac{\partial \lambda_1(\tth)}{\partial \tth}\Bigr\vert_{\tth=\tth_0}  {\alpha}( \bb Z_1,\bb v_1)
\right\}. 
\end{align*}
Note that $\bb\beta(u,\bb v_1)$ is finite for all $u\in[0,1]$ and $\bb v\in\R^d$ due to Assumptions \ref{as:g} and \ref{as:lambda}. 
The following theorem describes the asymptotic distribution of the proposed test statistic under the null hypothesis ${\cal H}_0$.

 \begin{thm}\label{th1}
Let Assumptions \ref{as:h}--\ref{as:W} be satisfied. Then the test statistic $\Delta_{T,W}$  converges in distribution for $T\to\infty$ to a random variable
\[
\int_{0}^1\int_{\R^d} |\mathcal{Z}(u,\bb v)|^2 W(u,\bb v) {\rm{d}}\bb v {\rm{d}} u,
\]
where $\mathcal{Z}=\{\mathcal{Z}(u,\bb v), s\in[0,1], \bb v \in\R^d\}$ is a centered Gaussian process with a covariance function
\begin{align*}
\E\left\{\mathcal{Z}(u_1,\bb v_1),\mathcal{Z}(u_2,\bb v_2)\right\}=&\E\left[\left(u_1^{Y_1}-g_{\lambda_1}(u_1) \right)\alpha(\tZ_1,\bb v_1) + \bb s_1 (\tth_0, Y_1, \mathcal{I}_{0})^\top \bb\beta(u_1,\bb v_1)\right] \\
&\times \left[\left(u_2^{Y_1}-g_{\lambda_1}(u_2) \right)\alpha(\tZ_1,\bb v_2) + \bb s_1 (\tth_0, Y_1, \mathcal{I}_{0})^\top \bb\beta(u_2,\bb v_2)\right].
\end{align*}
\end{thm}

 The assertion provides an  approximation of the  distribution of  $\Delta_{T,W}$ under the null hypothesis.  It follows from Theorem~\ref{th1} that $\Delta_{T,W}$ is asymptotically distributed as an infinite weighted sum of $\chi^2_1$ distributed random variable, where, however, the weights depend on the unknown parameters via a highly non-trivial way.  
 % Unfortunately,  
 Hence, it is extremely complicated to use this asymptotic distribution %depends on  unknown quantities and 
 in order to determine critical points and actually carry out the test. Instead, in Section~\ref{sec:boot} 
 we recommend the use of a bootstrap test which is straightforward to apply.
 % and simulations confirm  reasonable behavior for  some particular situation.

\subsection{Behavior under the null  for an ARX type of model}\label{sec:asAR}
 
An important special case of the model \eqref{rate2} is obtained for $q=0$, when the conditional mean  $\lambda_t$ depends only on the past values of $Y_t$ and on $\bb X_{t-1}$ via some kind of an ARX structure. Some of the assumptions specified in the previous section can be weakened for these models, so we treat this situation in more detail. If $q=0$, then
\[
h(\bb y, \bb x;\tth) = r(\bb y;\tth)+\pi(\bb x;\tth), \quad \bb y\in[0,\infty)^p, \bb x\in \R^m, \tth\in\R^K
\]
for some functions $r$ and $\pi$, and 
$\bb Z_t=(Y_{t-1},\dots,Y_{t-q},\bb X_{t-1}^\top)^\top$ depends solely on past values of $Y_t$ and on $\bb X_{t-1}$. This fact enables us to reduce and simplify the assumptions required in \ref{as:r}, \ref{as:lambda} and~\ref{as:W}, into the following weaker versions:

\setcounter{assumpT}{2}
\begin{assumpT}\label{as:r-AR}
 Let $r:[0,\infty)^p\times \R^K \to (0,\infty)$ and $\pi:\R^{K+m}\to(0,\infty)$ be functions such that $r(\bb y;\cdot\dot)$ and $\pi(\bb x;\cdot)$ are twice continuously differentiable for all $\bb y\in[0,\infty)^p$ and $\bb x\in\R^m$. Assume that there exists a neighborhood $\mathcal{V}(\tth_0)$ of $\tth_0$ such that % the matrices
 \begin{align*}
 %\E \sup_{\tth\in\mathcal{V}(\tth_0)} \left|r_{\tth}(\bb Y_{t-1:t-p};\tth)r_{\tth}(\bb Y_{t-1:t-p};\tth)^\top   \right|,  \E \sup_{\tth\in\mathcal{V}(\tth_0)} \left|\pi_{\tth}(\bb X_{t-1};\tth)\pi_{\tth}(\bb X_{t-1};\tth)^\top   \right|,
 %\E \sup_{\tth\in\mathcal{V}(\tth_0)} \left|r_{\tth}(\bb Y_{t-1:t-p};\tth)r_{\tth}(\bb Y_{t-1:t-p};\tth)^\top   \right|
 \E& \sup_{\tth\in\mathcal{V}(\tth_0)} \left| \frac{\partial r(\bb Y_{t-1:t-p};\tth)}{\partial \theta_i} \frac{\partial r(\bb Y_{t-1:t-p};\tth)}{\partial \theta_j}\right|<\infty, \quad
 \E \sup_{\tth\in\mathcal{V}(\tth_0)} \left| \frac{\partial^2 r(\bb Y_{t-1:t-p};\tth)}{\partial \theta_i \partial \theta_j} \right|<\infty, \\
 \E & \sup_{\tth\in\mathcal{V}(\tth_0)} \left| \frac{\partial \pi(\bb X_{t-1};\tth)}{\partial \theta_i} \frac{\partial \pi(\bb X_{t-1};\tth)}{\partial \theta_j}\right|<\infty, \quad
     \E \sup_{\tth\in\mathcal{V}(\tth_0)} \left| \frac{\partial^2 \pi(\bb X_{t-1};\tth)}{\partial \theta_i\partial \theta_j} \right|<\infty
 \end{align*}
 for all $i,j=1,\dots,K$, where $\bb Y_{t-1:t-p}$ is defined below \eqref{eq:Esup} and $Y_t,\tX_t$ are the stationary variables from Assumption~\ref{as:stat}.
  \end{assumpT}

\setcounter{assumpT}{4}
\begin{assumpT}\label{as:lambda-AR}
Assume that 
\[
\E H^2(\lambda_t)<\infty, 
\]
where the function $H(\cdot)$ is from Assumption~\ref{as:g} and $(Y_t,\lambda_t)$   are the stationary variables from Assumption~\ref{as:stat}.
\end{assumpT}

%Namely, if $h(\bb y,\bb l,\bb x;\tth)$ does not depend on $\bb l$ and we take $\bb Z_t = (Y_{t-1},\dots,Y_{t-p},\bb X_{t-1}^\top)^\top$, then the following set of assumptions is satisfactory:

\setcounter{assumpT}{6}
\begin{assumpT}\label{as:W-AR}
Let  
$W:\R^{d+1}\to(0,1)$ be  $W(u,\bb v)=w(u)\omega(\bb v)$ for some measurable %non-negative 
functions $w:(0,1)\to (0,\infty)$ and $\omega:\R^d\to(0,\infty)$ 
such that $\omega(\bb v)= \omega(-\bb v)$ for all %$u\in[0,1], 
$\bb v \in\mathbb R^d$  and %it holds $w(u)>0$, $\omega(\bb v)>0$,
\[
0<\int_{0}^1 w(u) \mathrm{d} u<\infty, \quad 0<\int_{\R^d} \omega(\bb v) \mathrm{d}\bb v<\infty.
\]
\end{assumpT}

 \begin{thm}\label{th1-AR}
 Let assumptions \ref{as:h}--\ref{as:stat} hold for $q=0$, and \ref{as:r-AR}, \ref{as:g}, \ref{as:lambda-AR}, \ref{as:est} and \ref{as:W-AR} be satisfied. 
 Then the test statistic $\Delta_{T,W}$  converges in distribution for $T\to\infty$ to a random variable
$
\int_{0}^1\int_{\R^d} |\mathcal{Z}(u,\bb v)|^2 W(u,\bb v) {\rm{d}}\bb v {\rm{d}} u,
$
where $\mathcal{Z}$ is the process specified in Theorem~\ref{th1}. 
\end{thm}

\subsection{Behavior under some alternatives}\label{sec:asALT}
 
The limit behavior of the test statistic is provided also for some types of alternatives. To this end analogously to the notation below \eqref{eq:Esup}, we write  $\bb Y_{k:l}$ for $(Y_k,Y_{k-1},\dots,Y_{k-(k-l)})^\top$ and $\bb \Lambda_{k:l}$ for $(\lambda_k,\lambda_{k-1},\dots,\lambda_{k-(k-l)})^\top$,  for any integers $(k,l)$ with $k\geq l$.

Let $p_0>0$, $q_0\geq 0$, $K_0\geq 1$ and write $h_0:[0,\infty)^{p_0+q_0}\times \R^{m+K_0}\to(0,\infty)$ for a  particular link function. % function satisfying Assumption
Consider the null hypothesis $\mathcal{H}_0$  whereby the link function is specified by $h_0$ and  a fixed family  of distributions $\{\mathcal{F}_\lambda,\lambda>0\}$, that is
\begin{align*} %\label{eq:mH0}
\mathcal{H}_0: &\text{ there exists } \tth\in \mathsf{int}(\bb\Theta_0) \text{ such that } Y_t|\mathcal{I}_{t-1}\sim \mathcal{F}_{\lambda_t}  \notag\\
&   \lambda_t = h_0(\bb Y_{t-1:t-p_0},\bb \Lambda_{t-1:t-q_0},\bb X_{t-1};\tth),
\end{align*}
for $\bb \Theta_0\subset \bb \Theta$ a compact set. A general alternative admits an incorrect specification of the conditional mean function as well as the conditional distribution. Specifically suppose that the actual model is given by 
\begin{equation*} %\label{eq:alt}
Y_t|\mathcal{I}_{t-1}\sim \mathcal{F}^A_{\lambda_t} \quad \lambda_{t} = h_A(\bb Y_{t-1:t-p_A},\bb \Lambda_{t-1:t-q_A},\bb X_{t-1};\tth_A)
\end{equation*}
 for some family of distributions $\{\mathcal{F}_{\lambda}^A, \lambda>0\}$, 
a link function $h_A:[0,\infty)^{p_A+q_A}\times \R^{m+K_A}\to(0,\infty)$, with  $p_A>0$, $q_A\geq 0$, $K_A\geq 1$, and for some $\tth_A\in \R^{K_A}$.  %$\tth_A\in\mathsf{int}(\bb \Theta_A)\subset \R^{K_A}$. 
Hence if the true pair $(\mathcal{F}^A_{\lambda},h^A)$ differs from $(\mathcal{F}_{\lambda},h_0)$
%If $\mathcal{F}_{\lambda}\ne \mathcal{F}_{\lambda}^A$ or $h_0\ne h_A$,
 then the model is not correctly specified under $\mathcal{H}_0$. 
 Let $\Delta_{T,W}$ be the 
  test statistic defined by  \eqref{teststCM-gen} for testing  $\mathcal{H}_0$,  based on 
  some estimator $\widehat{\tth}_T$ computed under $\mathcal{H}_0$, 
  with the sequence $\{\widehat{\lambda}_{t,0}\}$ defined recursively 
  using the link function $h_0$, and 
a $d$-dimensional sequence  $\{\widehat{\bb Z}_{t,0}\}$ where
  $d=p_0+q_0+m$.   
 Assume that the PGF $g_\lambda(u)$ of $\mathcal{F}_\lambda$ satisfies  
 \begin{equation}\label{eq:as_g_H1}
 |\partial g_\lambda(u)/\partial \lambda|<M_1 \text{ for all } u\in[0,1] \text{ and all } \lambda\in(0,\infty)
 \end{equation} 
for some $M_1>0$ and that there exists  $\tth_0\in\mathsf{int}(\bb \Theta_0)$ such that
\begin{equation}\label{eq:conv.th.alt}
\widehat{\bb \theta}_T \Pto \bb \tth_0, \quad T\to\infty.
\end{equation}

Consider first an alternative whereby the mean function $h$ is  correctly specified, but  the conditional distribution is different from the hypothesized one, i.e. 
$\mathcal{F}_{\lambda}\ne \mathcal{F}_{\lambda}^A$, and write $g_{\lambda}^A$ for the PGF corresponding to $\mathcal{F}_{\lambda}^A$. 
Note that in  this case the true parameter $\tth_0$ can still be  consistently estimated using, for instance, the Poisson QMLE, and accordingly we may assume that \eqref{eq:conv.th.alt} holds with $\tth_0$ and $\tth_A$ being identical.

 \begin{thm}\label{th3}
 Let $\{Y_t, \lambda_t\}$  and $\{\tX_t\}$ satisfy \ref{as:stat} with the conditional distribution $\mathcal{F}_{\lambda}^A$ and the link function $h_0$, and let 
  \ref{as:h}, \ref{as:r}, \ref{as:lambda}, \ref{as:W}, \eqref{eq:as_g_H1} and  \eqref{eq:conv.th.alt} hold. 
%  Let $g(u,z)=g_z(u)$ be differentiable on $(0,1)\times (0,\infty)$ such that $|\partial g(u,z)/\partial z|\leq M_1$ for some finite constant $M_1$. 
  Then 
 \[
 \frac{ \Delta_{T,W}}{T} \Pto \int_{0}^1\int_{\R^d}  \zeta^2(u,\bb v) W(u,\bb v) \mathrm{d}\bb v\mathrm{d} u,
\]
as $T \to \infty$, where
\[
\zeta(u,\bb v)=  \E\Big\{\big[g^A_{\lambda_1}(u) - g_{\lambda_1}(u)\big]  \alpha(\bb Z_1,\bb v)\Big\}
\]
and %$\bb Z_t = (Y_{t-1}, \dots, Y_{t-p_0}, \lambda_{t-1},\dots,\lambda_{t-q_0},\bb X_{t-1}^\top)$. 
$\bb Z_t = (\bb Y_{t-1:t-p_0}^\top, \bb\Lambda_{t-1:t-q_0}^\top,\bb X_{t-1}^\top)$. 
\end{thm}

% Note that  \ref{as:lambda} can be weakened in Theorem~\ref{th3} and  it suffices to assume finiteness of the first moment on a neighbourhood of $\bb \theta_0$. 

\bigskip
Consider now a situation whereby the conditional distribution is correctly specified   under the null hypothesis but the link function is misspecified, i.e. $\mathcal{F}^A_{\lambda}= \mathcal{F}_{\lambda}$ but $h_A\ne h_0$. 
We provide the asymptotic behavior of $\Delta_{T,W}$ in two practically important situations, when
either $h_0$ corresponds to  a linear model, or when $h_0$ is non-linear
with
 %$h_0\ne h_A$, i.e. the link function is not well specified, and the link $h_0$ corresponds to an ARX model, i.e. 
$q_0=0$, i.e. the tested link $h_0$ corresponds to an ARX type of a model. For both situations we show that the test based on $\Delta_{T,W}$ is consistent.

Let us start with the latter (non--linear) case, whereby
 $h_0(\bb y,\bb x;\tth)=r_0(\bb y;\tth)+\pi_0(\bb x;\tth)$ for some 
 $r:[0,\infty)^{p_0}\times \R^{K}\to(0,\infty)$ and $\pi: \R^{K_0+m}\to(0,\infty)$ and
  \[
     \widehat{\bb Z}_{t,0} = {\bb Z}_{t,0} = (\bb Y_{t-1:t-p_0}^\top,\bb X_{t-1}^\top)^\top
  \]
  for $t\geq p_0+1$. Notice that nothing is assumed about the true orders $(p_A,q_A)$, so these can be completely arbitrary. 
  %In the following we use notation $g(u,\lambda) = g_{\lambda}(u)$ for the PGF of $\mathcal{F}_{\lambda}$. 

\begin{thm} \label{th4}
Let $h_0:[0,\infty)^{p_0}\times \R^{K_0+m}\to(0,\infty)$ be a function satisfying  \ref{as:h} on $\bb \Theta\subset\R^{K_0}$  for some  $p_0>0$, $K_0\geq 1$.
Let  $\{(Y_t,\lambda_{t,A})\}$ and  $\{\bb X_t\}$ satisfy \ref{as:stat} with a link function $h_A$ and distribution $\mathcal{F}_\lambda$. 
% Let  $\{\bb X_t\}$ be stationary and ergodic and let $\{(Y_t,\lambda_t)\}$  be the  stationary and ergodic solution of \eqref{eq:alt}. 
 Assume  that  \eqref{eq:as_g_H1}, \eqref{eq:conv.th.alt} and \ref{as:W-AR}
 hold, and 
%$
%\widehat{\tth}_T\Pto \tth_0
%$
%for some $\tth_0\in\mathsf{int}(\bb \Theta)$. %, where $\bb \Theta_0\subset \bb \Theta$ is a compact subset.
% Let %he functions $r_0$ and $\pi_0$ satisfy that 
  there exists a neighborhood $\mathcal{V}(\tth_0)$ of $\tth_0$ such that
  \[
% \begin{align*}
% \E\left| \frac{\partial r_0(\bb Y_{t-1:t-p_0}; \tth) }{\partial\tth}\vert_{\tth=\tth_0}\right|<\infty,\quad  
% \E\left| \frac{\partial \pi_0(\bb X_{t-1}; \tth) }{\partial\tth}\vert_{\tth=\tth_0}\right|<\infty.  
\E\sup_{\tth \in \mathcal{V}(\tth_0)} \left|  h_{0}(\bb Y_{t-1:t-p_0},\bb X_{t-1}; \tth) \right|<\infty,\quad  
\E\sup_{\tth \in \mathcal{V}(\tth_0)} \left| \frac{ \partial h_{0}(\bb Y_{t-1:t-p_0},\bb X_{t-1}; \tth)}{\partial \tth} \right|<\infty.
%\E\sup_{\tth \in \mathcal{V}(\tth_0)} \left|  r_{0\tth}(\bb Y_{t-1:t-p_0}; \tth_0) \right|&<\infty,\quad  
%\E\sup_{\tth \in \mathcal{V}(\tth_0)} \left|  \pi_{0\tth}(\bb X_{t-1}; \tth_0) \right|<\infty.  
 %\end{align*}
 \]
%Assume that there exist $M_1>0$ such that  $|\partial g_z(u)/\partial z|<M_1$ for all $u\in[0,1]$ and all $z\in(0,\infty)$. 
  Then
  \[
  \frac{1}{T}\Delta_{T,W}\Pto \int_{0}^1 \int_{\R^d} \zeta^2(u,\bb v) w(u)\omega(\bb v )\mathrm{d}\bb v \mathrm{d}u
  \]
 as $T\to\infty$, where
 % \[
  %\xi(u,\bb v) = \E\left\{\left[ g(u,\lambda_1)-g(u,h_0(\bb Y_{0:1-p_0},\bb X_{0};\tth_0))\right]\alpha({\bb Z}_{1},\bb v)\right\}.
  %\]
   \[
  \zeta(u,\bb v) = \E\big\{\left[ g_{\lambda_{1,A}}(u) -g_{\lambda_{1,0}}(u)\right]\alpha({\bb Z}_{1,0},\bb v)\big\}, \quad \lambda_{1,0}=h_0(\bb Y_{0:1-p_0},\bb X_{0};\tth_0).
  \]
  %for $l_1=h_0(\bb Y_{0:1-p_0},\bb X_{0};\tth_0)$. 
 % and
  %\[
   %{\bb Z}_{t} = (\bb Y_{t-1:t-p_0},\bb X_{t-1}^\top)^\top.
 % \]
  \end{thm}

 Consider now a situation, where $h_0$ corresponds to a linear model with orders $p=q=1$, that is
 \begin{equation}\label{eq:h0_lin}
 h_0(y,\lambda,\tx,;\tth) = \omega+\alpha y+\beta \lambda +\tgamma^\top\tx
 \end{equation}
 for $\tth=(\omega,\alpha,\beta,\tgamma^\top)^\top$. 
  Let $\bb \Theta$ be a subset of $\R^{3+m}$ such that for any $\tth\in\bb\Theta$, $\tth=(\omega,\alpha,\beta,\tgamma^\top)^\top$, it holds that $\omega,\alpha,\beta>0$ and $\alpha+\beta<1$. Again, nothing is assumed about the true orders $(p_A,q_A)$ and only the standard assumptions are posed for the true link function $h_A$.

\begin{thm}\label{th5}
 Let  $\{Y_t, \lambda_{t,A}\}$  and $\{\bb X_t\}$ satisfy \ref{as:stat} with a link function $h_A$ and distribution $\mathcal{F}_{\lambda}$. Let $h_0$ be given by \eqref{eq:h0_lin} and $\bb \Theta_0\subset\bb \Theta$ a compact let.
 Assume that \eqref{eq:as_g_H1}, \eqref{eq:conv.th.alt} and  \ref{as:W}
 hold. 
Then
\[
\frac{1}{T}\Delta_{T,W} \Pto  \int_0^1\int_{\R^d} \zeta^2(u,\bb v) w(u) \omega(\bb v) \mathrm{d}\bb v \mathrm{d} u
\]
 as $T\to\infty$, where
\[
\zeta(u,\bb v) =  \E\big\{ \left[g_{\lambda_{1,A}}(u) - g_{\lambda_{1,0}(\tth_0)}(u)\right]\alpha\bigl({\bb Z}_{1,0}(\tth_0),\bb v\bigr)\big\}, 
\]
%\[
%\zeta(u,\bb v) =  \E\big\{ \left[g(u;\lambda_{1,A}) - %g(u;\lambda_{1,0}(\tth_0))\right]\alpha\bigl({\bb Z}_{1,0}%(\tth_0),\bb v\bigr)\big\}, 
%\]
and $\{\lambda_{t,0}(\tth_0))\}$ is a stationary sequence such that
\[
{\lambda}_{t,0}(\tth)= 
\frac{\omega}{1-\beta} +\sum_{j=0}^{\infty} \beta^j [\alpha Y_{t-1-j}+\tgamma^\top\tX_{t-1-j}],
\]
and ${\bb Z}_{t,0}(\tth)=(Y_{t-1},\lambda_{t-1,0}(\tth),\bb X_{t-1}^\top)^\top$. 
\end{thm}

In order to further discuss the conditions of consistency, write $\Delta^\zeta_{\infty,W}$ for each probability limit of $\Delta_{T,W}/T$ figuring in Theorems \ref{th3}, \ref{th4} and \ref{th5}.  Since $\zeta$ is a continuous function on $(0,1)\times \R^d$,  $\Delta^\zeta_{\infty,W}=0$ only if  $\zeta(u,\bb v)$ is equal to zero identically in $(u,\bb v) \in (0,1)\times \R^d$. On the other hand if $\zeta(u,\bb v)\neq 0$ for some $(u,\bb v) \in (0,1)\times \R^d$, then it follows from assumption \ref{as:W} that 
%then 
$\Delta^\zeta_{\infty,W}>0$, and consequently the corresponding test that rejects the null hypothesis  ${\cal{H}}_0$ 
for large values of $\Delta_{T,W}$ is consistent, in the sense that the rejection probability is equal to one asymptotically under the violations of ${\cal{H}}_0$ considered herein.

In this connection note that in Theorem \ref{th3} $g^A_{\lambda}\neq  g_{\lambda}$, and therefore 
$g_{\lambda_1}^A(u)-g_{\lambda_1}(u)$ is generally a non-zero random variable. If in addition $g^A_{\lambda}(u)>g_{\lambda}(u)$ for all $u\in(0,1)$ (which is the case for $\mathcal{F}_{\lambda}$ Poisson and $\mathcal{F}^A_{\lambda}$ a Negative Binomial), then $g_{\lambda_1}^A(u)-g_{\lambda_1}(u)>0$ for all $u\in(0,1)$. 
In the case of Theorems \ref{th4} and \ref{th5},
 if the condition~\eqref{eq:SOC} holds sharply in the sense that $\{\lambda_1 <  \lambda_2 \implies g_{\lambda_1}(u)>g_{\lambda_2}(u) \text{ for all } \ u\in(0,1) \}$  (as, e.g., for the Poisson distribution), then  $\lambda_{1,A} - \lambda_{1,0} \ne 0$ a.s. implies that 
$g_{\lambda_{1,A}}(u)-g_{\lambda_{1,0}}(u)\ne 0$ a.s. for all $u\in(0,1)$.
  Also observe that $\alpha(\bb v_1,\bb v_2)=0$ only if $\bb v_1^\top \bb v_2=(3\pi/4)+k\pi, \ k=0,\pm1,\pm2,... \ $. 
  Hence in each of the three cases of alternatives considered by Theorems \ref{th3}--\ref{th5}, the random function within the corresponding expectation in the definition of $\zeta(u,\bb v)$ is generally non-zero. 
  This of course does not rule out the case of $\zeta(\cdot,\cdot)$ being identically equal to zero, but alludes to the very special circumstances under which the test based on $\Delta_{T,W}$ will miss any of the alternatives considered herein, at least for large sample size $T$.   In this connection, the small sample behavior of the test based on $\Delta_{T,W}$ under various realizations of the alternatives considered in Theorems \ref{th3}--\ref{th5} is investigated by means of  a Monte Carlo simulation study in Section~\ref{sec:simul} and shows that the new test has non--trivial power against several instances of violations of the null hypothesis.

%In this connection note that the conditions of Theorem \ref{th3} imply that $g^A_{\lambda}\neq  g_{\lambda}$, while an identifiability condition $\{\lambda_1\neq %\lambda_2 \implies \mathcal F_{\lambda_1}(\cdot)\neq \mathcal F_{\lambda_2}(\cdot)\}$, and since $\lambda_{1,A} \neq \lambda_{1,0}$, we have $g_{\lambda_{1,A}}-g_{\lambda_{1,0}}\neq 0$ in the case of Theorems \ref{th4} and \ref{th5}. Also observe that $\alpha(\bb v_1,\bb v_2)=0$ only if $\bb v_1^\top v_2=(3\pi/4)+k\pi, \ k=0,\pm1,\pm2,... \ $. Hence in each of the three cases of alternatives considered by Theorems \ref{th3}--\ref{th5}, the function within the corresponding expectation in the definition of $\zeta(u,\bb v)$ is different from zero. This of course does not rule out the case of $\zeta(\cdot,\cdot)$ being identically equal to zero, but alludes to the very special circumstances under which the test based on $\Delta_{T,W}$ will miss any of the alternatives considered herein, at least for large sample size $T$.   In this connection, the small sample behavior of the test based on $\Delta_{T,W}$ under various realizations of the alternatives considered in Theorems \ref{th3}--\ref{th5}, is investigated by means of  a Monte Carlo simulation study in Section~\ref{sec:simul} and shows that the new test has non--trivial power against several instances of violations of the null hypothesis.

\section{Practical computation and bootstrap test}\label{sec:boot}

 \subsection{Analytic computations}

Let the weight function $W$ %satisfy  Assumption~\ref{as:W}, that is
be of the form  $W(u,\bb v)=w(u)\omega(\bb v)$. Due to \eqref{Dnsum}
 the test statistic $\Delta_{T,W}$  takes the form
\begin{equation} \label{testsum}
\Delta_{T,W}= \frac{1}{T} \sum_{t,s=1}^T K_{w,ts} \widetilde{K}_{\omega,ts},
\end{equation}
with 
\[
K_{w,ts} = \int_{0}^1 \widehat{\eps}_t(u)\widehat{\eps}_s(u) w(u)\mathrm{d}u= \int_{0}^1 \left[u^{Y_t} - g_{\widehat{\lambda_t}}(u)\right]\left[u^{Y_s} - g_{\widehat{\lambda_s}}(u)\right] w(u)\mathrm{d}u,
\]
and
%$K_{w,ts}:=K_w({\widehat {\varepsilon}}_t,{\widehat {\varepsilon}}_s)$ and
 $\widetilde{K}_{\omega,ts}:=\widetilde{K}_{\omega}(\bb {\widehat {Z}}_t-\bb {\widehat {Z}}_s)$,
where
\begin{equation*} % \label{int}
%K_w(\varepsilon_1,\varepsilon_2)=\int_{0}^1 \varepsilon_1(u) \varepsilon_2(u) w(u) {\rm{d}}u, \ \ \ 
\widetilde{K}_\omega(\bb z)=\int_{\R^d} \cos(\bb v^\top \bb z) \omega(\bb v) {\rm{d}}\bb v.
\end{equation*}
 The integral $\widetilde{K}_\omega(\cdot)$ %presents us with no computational difficulties as 
can be easily   analytically computed if the  weight function $\omega(\cdot)$ corresponds to the density of a spherical distribution. Then we immediately have $\widetilde{K}_\omega(\bb z)=\Xi(\|\bb z\|)$ with the function $\Xi(\cdot)$ being the generator of the characteristic function of the chosen spherical distribution. Hereafter we use $\Xi:=\Xi_{\gamma,\eta}$, with 
\begin{equation*} %\label{weight1}
\Xi_{\gamma,\eta}(\xi)=e^{-\gamma\xi^\eta}, \ \xi,\gamma>0, \ \eta \in (0,2],
\end{equation*} 
corresponding to the spherical stable characteristic function.  Notice that Assumption~\ref{as:W-AR} is satisfied for all $\gamma>0$ and $\eta \in (0,2],$, while only $\eta=2$ fulfills the requirements of Assumption~\ref{as:W}.

The integral $K_{w,ts}$ can be further simplified depending on the PGF $g_{\lambda}$ and the choice of $w(u)$. Commonly used function is
 $w_{\rho}(u)=u^{\rho}$ for $\rho\geq 0$. Another possibility is to take a weight function proportional to a density of a beta distribution
$w_{\rho,\kappa}(u)=u^{\rho}(1-u)^{\kappa}$ for $\rho,\kappa\geq 0$, which enables to regulate the weight given to neighborhoods of the origin and point $u=1$.  The case of linear Poisson PARX model and weight function $w_{\rho}$ is considered in more detail in Section~\ref{sec:PARX}.

 \subsection{Bootstrap version of the test}

  Theorem \ref{th1}  provides the asymptotic distribution of the test statistics $\Delta_{T,W}$  under the null hypothesis $\mathcal{H}_0$, but this distribution depends
on  several unknown quantities.  To get  an approximation for the desired critical value  one can   generate the Gaussian  process described in Theorem  \ref{th1} with the unknown quantities replaced by their estimators and then calculate the integral and its quantiles. However, here we propose  a parametric bootstrap.

For given data $(Y_1,\ldots,Y_T;\bb X_1,\ldots, \bb X_T)$ and specified weight function $W$ %tuning parameters $(\gamma,\eta,\rho)$ 
the test proceeds as follows:
\begin{enumerate}[(i)]
  \item  Calculate the  estimator $\widehat { \bb \theta}_T=\widehat { \bb \theta}_T  (Y_1,\ldots,Y_T;\bb X_1,\ldots, \bb X_T )$ of $\bb \theta$ and the test statistic $ \Delta_{W,T}=\Delta_{W,T} (Y_1,\ldots, Y_T;\bb X_1,\ldots, \bb X_T)$.
\item Generate bootstrap observations  $\bb X^*_1,\ldots,\bb X_T^*$ from  $\bb X_1,\ldots,\bb X_T$ by a block bootstrap. %  -add what kind block bootstrap.
\item Generate the bootstrap observations  $Y^*_1,\ldots, Y_T^*$ recursively %for given   $\bb X^*_1,\ldots,\bb X_T^*$
from the  distribution $\mathcal{F}_{ \widetilde {\lambda}^*_t(\widehat{\bb \theta}_T)}$ with mean
$ \widehat {\lambda}^*_t(\widehat{\bb \theta}_T)$ defined as
     $\widehat {\lambda}_t(\bb \theta)$
      with   $\bb \theta$ replaced by $\widehat {\bb \theta}_T$ and $Y_{t-i}$ replaced by $Y^*_{t-i}$ and $\bb X_t$ replaced by $\bb X_t^*$.
      % for some initial values for $\bb X_0$, $\lambda_0$ and $Y_0$. %\frac{1}{T} \sum_{t=1}^T Y_t$ .
     %  $\widetilde \lambda^*_1=\frac{1}{T} \sum_{t=1}^T Y_t$ .
\item Calculate the bootstrap estimate $\widehat{ \bb \theta}_T^*= \widehat { \bb \theta}_T  (Y^*_1,\ldots,Y^*_T;\bb X^*_1,\ldots, \bb X^*_T )$ of $\bb \theta$ and  the bootstrap test statistic $\Delta_{W,T}^*:=\Delta_{W,T}^* (Y^*_1,\ldots, Y^*_T;\bb X_1^*,\ldots, \bb X_T^*)$.
\item  Repeat steps (ii)--(iv) a number of times, say $B$, thereby producing the test statistics
     ($  \Delta_{W,T,b}^*, \ b=1,...,B$), where $\Delta_{W,T,b}^*= \Delta_{W,T,b}^* (Y^*_{1b},\ldots, Y_{Tb}^*,\bb X^*_{1 b},\ldots,\bb X_{Tb}^*)$.
\item Reject the null hypothesis if $ \Delta_{W,T}>  \Delta^*_{W,T, B- \lfloor B\alpha\rfloor}$, where 
     $ \Delta_{W,T,(1)}^*
      \leq \ldots   \leq  \Delta_{W,T,(B)}^*  $  stand for the order statistics corresponding to ($  \Delta_{W,T,b}^*, \ b=1,...,B$), and $\alpha$ denotes the prescribed level of significance.  
\end{enumerate}

The simulation study in next section uses overlapping block bootstrap with fixed block length $l$ with $l\approx n^{1/3}$.   The initial values for computation of $\widehat{\lambda}_t$ and $\widehat{\lambda}_t^*$   are chosen as    $\widehat \lambda_0=0$, $Y_0=Y_1$ and $\bb X_0=\bb X_1$.

\bigskip

Using similar tools as in \cite{Neu21}  while examining step--by--step the proof of Theorem \ref{th1}, one will come to the conclusion that  the  approximation of the critical value by the above parametric bootstrap is asymptotically correct if  the original data follow the null hypothesis. 
%If the data do not satisfy the null hypothesis then the estimator of   $\bb  \theta$ of course does not estimate the true parameter, but the test still remains asymptotically valid.

\section{Application to PARX model} \label{sec:PARX}

The linear Poisson autoregression with exogenous covariates (PARX) assumes that $\mathcal{F}_{\lambda}$ is Poisson, so the corresponding PGF takes form
$
g_\lambda(u)=e^{\lambda(u-1)},
$
 and the conditional mean is modelled as \eqref{eq:LIN} for some $p\geq 1$ and $q\geq 0$ and an unspecified model 
 parameter~$\tth$. %^\top=(\omega,\alpha_1,\dots,\alpha_p,\beta_1,\dots,\beta_q,\pi_1,\dots,\pi_m)$. 
 In the following this particular hypothesis is referred to as $\mathcal{H}_0^{\rm{P}}$. 
%Let $\widehat{\tth}$ be an % the ML 
%estimator of $\tth$. 
Then  
\[
\widehat \varepsilon_t(u)=u^{Y_{t}}-e^{\widehat{\lambda}_t(u-1)},  
\]
where $\widehat{\lambda}_t$ are computed recursively as described in Section~\ref{sec:gen}.  The integral $K_{w,ts}$ from \eqref{testsum} can be computed as
$K_{w,ts}:=K_w({\widehat {\varepsilon}}_t,{\widehat {\varepsilon}}_s)$ for 
%The integral $K_w(\varepsilon_1,\varepsilon_2)$ %recall that for the PARX model we have 
%$\lambda_t$ in the form \eqref{eq:LIN}  so the first integral in \eqref{int} specifies to
%can be written as
\begin{equation*}
K_w(\varepsilon_1,\varepsilon_2)=\int_{0}^1 (u^{y_{1}}-e^{\lambda_1(u-1)})  (u^{y_{2}}-e^{\lambda_2(u-1)}) w(u){\rm{d}}u,
\end{equation*} where  $\lambda_t$ and $y_t, \ t=1,2$, are used as generic notations for pairs of  parameters and associated responses, respectively.
To proceed further hereafter we adopt the weight function 
%$w:=w_\rho$ where
%\begin{equation} \label{weight}
$w(u)=w_\rho(u)=u^\rho$,  $\rho\geq 0$,
%\end{equation}
and write $K^{(\rho)}(\varepsilon_1,\varepsilon_2):=K^{(\rho)}(y_1,y_2;\lambda_1,\lambda_2)$,  for the resulting integral which after some straightforward algebra is rendered as
\begin{align*}
K^{(\rho)}(y_1,y_2;\lambda_1,\lambda_2)&=\frac{1}{1+y_1+y_2+\rho}\\
&- e^{-\lambda_2} I(y_1+\rho,\lambda_2)-e^{-\lambda_1} I(y_2+\rho,\lambda_1)+e^{-(\lambda_1+\lambda_2)} I(\rho,\lambda_1+\lambda_2),
\end{align*}
where
\begin{equation*} % \label{iota}
I(a,\theta)=\int_{0}^1 u^a e^{\theta u} {\rm{d}}u=\frac{(-\theta)^{-a}}{\theta} \left(\Gamma(1+a,-\theta)-\Gamma(1+a)\right).
\end{equation*}

Then the test statistic figuring in \eqref{testsum}  may be written as
\begin{equation*} %\label{testpoisson}
\Delta_{T,W}= \frac{1}{T} \sum_{t,s=1}^T \widehat {K}_{ts}^{(\rho)} \Xi_{\gamma,\eta}(\|\bb {\widehat Z}_t-\bb {\widehat Z}_s\|),
\end{equation*}
with $\widehat {K}_{ts}^{(\rho)}={K}^{(\rho)}(y_t, y_s;\widehat\lambda_t,\widehat\lambda_s)$ and $\widehat {\bb Z}_t$ defined
recursively as described in Section~\ref{sec:gen}. % under equation \eqref{eq:EPS}.

\section{Simulations}\label{sec:simul}

This simulation study explores the 
 finite sample behavior of the suggested bootstrap test  based on $\Delta_{T,W}$ for the null hypothesis  $\mathcal{H}_0^{\rm{P}}$, where the conditional distribution $\mathcal{F}_{\lambda}$ is Poisson. 
%The null hypothesis settings considered are the cases whereby under  ${\cal {H}}^{\rm P}_0$, 
The conditional mean is specified either as an ARX(1) kind of model
\begin{equation}
\lambda_t=\omega+\alpha_1 Y_{t-1} + c\left(\cos X_{t-1}+1\right), \label{eq:lam1}\tag{S1}
\end{equation}
or as a GARCHX(1,1) type of model
\begin{equation}
\lambda_t=\omega+\alpha_1 Y_{t-1} + 
\beta_1 \lambda_{t-1}+c\left(\cos X_{t-1}+1\right).\label{eq:lam2}\tag{S2}
\end{equation}
The exogenous series $\{X_t\}$ was generated from a stationary AR(1) model $X_t= \rho_X X_{t-1}+\eps_t$ for $\rho_X=0.5$ and $\{\eps_t\}$ iid centred normal with variance  $(1-\rho_X^2)^{-1}$. The model parameters were estimated by the maximum likelihood method. The performance of the testing procedure was investigated under the null hypothesis. Moreover the alternative hypothesis settings considered are the cases whereby we have:
\begin{enumerate}[(a)]
\item Incorrect model specification. %The data were generated from a model with different $\bb Z_t$. Namely, when testing 
Observations were generated  from a model of an incorrect  order, while the function $\pi$ from \eqref{eq:hNONLIN} and the conditional distribution $\mathcal{F}_{\lambda}$ were specified correctly. Namely,  the test for \eqref{eq:lam1} was conducted for observations from the GARCHX(1,1) model \eqref{eq:lam2} and from a ARX(2) kind of model
\begin{equation}
\lambda_t=\omega+\alpha_1 Y_{t-1} +\alpha_2Y_{t-2}+ c \left(\cos X_{t-1}+1\right),\label{eq:lam3}\tag{S3}
\end{equation}
which was used as
 an alternative for testing \eqref{eq:lam2} as well. 
\item Incorrect conditional distribution. Observations were generated from a model with a correctly specified conditional mean $\lambda_t$ but the conditional distribution was
 the Negative Binomial  $\mathsf{NB}(\lambda_t,r)$ with a dispersion parameter $r$. Recall that if a random variable follows $\mathsf{NB}(\lambda,r)$ then its mean is $\lambda$ and variance is $\lambda(1+\lambda/r)$. 
\item Incorrectly specified dependence on the exogenous variable $X_t$.  Observations were generated from a Poisson model with a correct orders $p,q$ and % either  \eqref{eq:lam1} or \eqref{eq:lam2} where 
 $\pi(X_{t-1};c)=c(\sin X_{t-1}+1)$ (abbreviated as ``sin" alternative) or 
$ \pi(X_{t-1};c)=c(2-|X_{t-1}|)I[|X_{t-1}|<2]$ (abbreviated as ``abs" alternative) and we test for  \eqref{eq:lam1} or \eqref{eq:lam2}, where $\pi(X_{t-1};c)=c(\cos X_{t-1}+1)$. 
\end{enumerate}
Results are presented for sample sizes $T=100,200,500$ and $T=1\,000$ at level of significance $\alpha=0.05$. 
The weight function $W(u,\boldsymbol{v})$ was considered as described in Section~\ref{sec:PARX}.
The tuning parameter $\rho$ was set as $\rho=0$, while the parameters
$\gamma,\eta$ were considered in the following combinations
\[
(\gamma,\eta)\in\Bigl\{
%\left(\frac14,\frac14\right),  \left(\frac12,\frac12\right),  \left(\frac12,1\right), \left(1,\frac12\right),(1,1),(1,2),(2,2,0)\Bigr\}.
\left(1/4,1/4\right),  \left(1/2,1/2\right),  \left(1/2,1\right), \left(1,1/2\right),(1,1),(1,2),(2,2)\Bigr\}.
\]
 For model (S1) all these settings satisfy the assumptions required by the weight function, but for (S2) we have theoretical justification only for the last two choices of $(\gamma,\eta)$ due to the stronger requirements in Assumption~\ref{as:W}. Despite this fact, we explore the behavior of the bootstrap test for all the choices of $(\gamma,\eta)$ listed above, so we can investigate how sensitive the  procedure is to violation of  Assumption~\ref{as:W}.

In addition, the behavior of the test statistic $\Delta_{T,W}$ was compared to tests based on the following two test  statistics:
\[
\Delta^{(0)}_{T}= T \int_{0}^1  \left| \frac{1}{T} \sum_{t=1}^T {\widehat {\varepsilon}^2}_t(u)\right|^2
 {\rm{d}}u, \qquad 
 % \label{teststCM1}
\Delta^{(1)}_{T}= \sqrt{T} \sup_{0<u<1}  \left| \frac{1}{T} \sum_{t=1}^T {\widehat {\varepsilon}}_t(u)\right|.
\]
The significance of $\Delta_{T,W}$ and  $\Delta^{(i)}_{T}$, $i=0,1$, was evaluated using the parametric  bootstrap described in Section \ref{sec:boot} with resampling size $B=499$. The simulations were conducted in the {\tt{R}}-computing environment  for $M=1\,000$ repetitions.

The size of the test %for $\alpha=0.05$ 
in Table~\ref{tab:size} is presented for the setting $\omega=0.2,$ $\alpha_1=0.3$, $\pi_1=0.5$ for \eqref{eq:lam1} and for $\omega=0.2$, $\alpha_1=0.3$, $\beta_1=0.3$, $c=0.5$ for \eqref{eq:lam2}. 
The results show that the test based on $\Delta_{T,W}$ and $\Delta^{(0)}_T$ keeps the prescribed confidence level. The test corresponding to $\Delta^{(1)}_{T}$ seems to be slightly conservative for \eqref{eq:lam1} model with $T<1\,000$.

\begin{table}[htbp]
\centering
\begin{tabular}{l|rrrrrrr|rr}
  \hline  \hline
&\multicolumn{9}{c}{${\cal{H}}^{\rm{P}}_0$: ARX(1)  model \eqref{eq:lam1}}\\
\hline
&\multicolumn{7}{c|}{$\Delta_{T,W}$ with $\rho=0$ and $(\gamma,\eta)$ }&\multicolumn{2}{c}{}\\
$T$ &  $(\frac14,\frac14)$  & $(\frac12,\frac12)$  & $(\frac12,1)$ & $(1,\frac12)$ & (1,1) & (1,2) &    (2,2) & $\Delta^{(0)}_{T}$ &$\Delta^{(1)}_{T}$ \\
 \hline
  100 & 0.054 & 0.050 & 0.040 & 0.041 & 0.046 & 0.045 & 0.046 & 0.052 & 0.027 \\ 
  200 & 0.036 & 0.043 & 0.053 & 0.054 & 0.057 & 0.052 & 0.061 & 0.038 & 0.037 \\ 
  500 & 0.047 & 0.050 & 0.049 & 0.048 & 0.055 & 0.054 & 0.057 & 0.048 & 0.025 \\ 
  1000 & 0.047 & 0.048 & 0.048 & 0.051 & 0.053 & 0.049 & 0.043 & 0.050 & 0.046 \\ 
   \hline  
&\multicolumn{9}{c}{${\cal{H}}^{\rm{P}}_0$: GARCHX(1,1)  model \eqref{eq:lam2}}\\
\hline
&\multicolumn{7}{c|}{$\Delta_{T,W}$ with $\rho=0$ and $(\gamma,\eta)$ }&\multicolumn{2}{c}{}\\
$T$ &  $(\frac14,\frac14)$  & $(\frac12,\frac12)$  & $(\frac12,1)$ & $(1,\frac12)$ & (1,1) & (1,2) &    (2,2) & $\Delta^{(0)}_{T}$ &$\Delta^{(1)}_{T}$ \\
  \hline
  100 & 0.040 & 0.041 & 0.037 & 0.042 & 0.037 & 0.039 & 0.039 & 0.041 & 0.032 \\ 
  200 & 0.051 & 0.043 & 0.038 & 0.041 & 0.035 & 0.028 & 0.030 & 0.053 & 0.046 \\ 
  500 & 0.049 & 0.050 & 0.048 & 0.046 & 0.039 & 0.049 & 0.049 & 0.054 & 0.055 \\ 
  1000 & 0.048 & 0.041 & 0.041 & 0.039 & 0.040 & 0.040 & 0.041 & 0.041 & 0.055 \\ 
   \hline \hline
\end{tabular}\caption{Size of the bootstrap test for $\lambda_t$ specified by \eqref{eq:lam1} and \eqref{eq:lam2} for $\alpha=0.05$.} \label{tab:size}
\end{table}

 \begin{table}[htbp]
\centering 
\begin{tabular}{ll|rrrrrrr|rr}
 \hline\hline
&&\multicolumn{7}{c|}{$\Delta_{T,W}$ with $\rho=0$ and $(\gamma,\eta)$ }&\multicolumn{2}{c}{}\\&
 Alt. &  $(\frac14,\frac14)$  & $(\frac12,\frac12)$  & $(\frac12,1)$ & $(1,\frac12)$ & (1,1) & (1,2) & (2,2) & $\Delta^{(0)}_{T}$ &$\Delta^{(1)}_{T}$ \\
 \hline
\multirow{4}{*}{\rot{ARX(2)}}  & 100 & 0.485 & 0.462 & 0.439 & 0.397 & 0.379 & 0.346 & 0.318 & 0.365 & 0.446 \\ 
   & 200 & 0.772 & 0.729 & 0.684 & 0.643 & 0.595 & 0.554 & 0.518 & 0.687 & 0.767 \\ 
   & 500 & 0.998 & 0.989 & 0.969 & 0.967 & 0.934 & 0.886 & 0.859 & 0.985 & 1.000 \\ 
   & 1000 & 1.000 & 1.000 & 0.999 & 1.000 & 0.998 & 0.991 & 0.990 & 1.000 & 1.000 \\ 
   \hline
\multirow{4}{*}{\rot{GARCHX}}  & 100 & 0.102 & 0.097 & 0.088 & 0.081 & 0.070 & 0.065 & 0.062 & 0.045 & 0.061 \\ 
   & 200 & 0.154 & 0.158 & 0.146 & 0.128 & 0.116 & 0.113 & 0.104 & 0.116 & 0.137 \\ 
   & 500 & 0.310 & 0.309 & 0.259 & 0.251 & 0.219 & 0.182 & 0.171 & 0.326 & 0.317 \\ 
   & 1000 & 0.576 & 0.539 & 0.386 & 0.440 & 0.314 & 0.242 & 0.226 & 0.534 & 0.600 \\ 
   \hline
\multirow{4}{*}{\rot{NB}}  & 100 & 0.639 & 0.632 & 0.575 & 0.609 & 0.460 & 0.318 & 0.263 & 0.618 & 0.522 \\ 
   & 200 & 0.889 & 0.884 & 0.841 & 0.858 & 0.705 & 0.532 & 0.408 & 0.881 & 0.812 \\ 
   & 500 & 1.000 & 1.000 & 0.999 & 0.999 & 0.987 & 0.911 & 0.795 & 1.000 & 0.995 \\ 
   & 1000 & 1.000 & 1.000 & 1.000 & 1.000 & 1.000 & 0.999 & 0.989 & 1.000 & 1.000 \\ 
   \hline
\multirow{4}{*}{\rot{sin}} & 100 & 0.153 & 0.534 & 0.728 & 0.675 & 0.738 & 0.724 & 0.594 & 0.099 & 0.089 \\ 
   & 200 & 0.343 & 0.893 & 0.972 & 0.953 & 0.969 & 0.968 & 0.914 & 0.163 & 0.145 \\ 
   & 500 & 0.887 & 1.000 & 1.000 & 1.000 & 1.000 & 1.000 & 1.000 & 0.292 & 0.283 \\ 
   & 1000 & 1.000 & 1.000 & 1.000 & 1.000 & 1.000 & 1.000 & 1.000 & 0.608 & 0.546 \\ 
   \hline
\multirow{4}{*}{\rot{abs}}  & 100 & 0.058 & 0.054 & 0.053 & 0.057 & 0.064 & 0.059 & 0.061 & 0.049 & 0.047 \\ 
   & 200 & 0.058 & 0.062 & 0.077 & 0.084 & 0.091 & 0.101 & 0.103 & 0.054 & 0.054 \\ 
   & 500 & 0.076 & 0.091 & 0.091 & 0.110 & 0.126 & 0.116 & 0.130 & 0.071 & 0.065 \\ 
   & 1000 & 0.095 & 0.110 & 0.128 & 0.162 & 0.190 & 0.232 & 0.235 & 0.103 & 0.087 \\ 
  \hline\hline
  \end{tabular}\caption{Power of test for ARX(1) model \eqref{eq:lam1} for data generated from ARX(2) model \eqref{eq:lam3}, data from GARCHX(1,1) model \eqref{eq:lam2},
 data from the Negative Binomial distribution with $r=3$ (NB), and data from incorrectly specified link part $\pi$, where the data are generated using $\pi(x;c)=c(\sin x+1)$ (sin) or $\pi(x;c)=c (2-|x|)I[|x|<2]$ (abs), while we test $\pi(x;c)=c(\cos x +1)$.}\label{tab:AR-H1}
\end{table}
 
The power of the test for \eqref{eq:lam1} in Table~\ref{tab:AR-H1} was computed for data generated from a ARX(2) type of model \eqref{eq:lam3} with  $\omega=0.2,$ $\alpha_1=0.3$, $\alpha_2=0.6$, $c=0.5$,  
a GARCHX(1,1)  model \eqref{eq:lam2} with  $\omega=0.2,$ $\alpha_1=0.3$, $\beta_1=0.6$, $c=0.5$,  a Negative Binomial model with $r=3$ and $\lambda_t$ given by \eqref{eq:lam1}
 for  $\omega=0.2,$ $\alpha_1=0.3$, $c=0.5$, and for a model with misspecified $\pi$ with $\omega=0.2,$ $\alpha_1=0.3$, $c=0.5$. 
 
A first observation is that the test based on $\Delta_{T,W}$ clearly outperforms its competitors for the ``$\sin$" misspecification alternative of the $\pi$ function regardless of the choice of  the tuning parameters $\gamma$ and $\eta$. Otherwise the new test appears to be often sensitive to the choice of $\gamma$ and $\eta$, with  the choices   $(\gamma,\eta)=(1/4,1/4)$, and    $(\gamma,\eta)=(1/2,1/2)$ being overall preferable. In this connection,  the test based on  $(\gamma,\eta)=(1/4,1/4)$, 
outperforms the tests based on $\Delta^{(0)}_{T}$, and   $\Delta^{(1)}_{T}$, nearly uniformly over the sample size $T$ and against all alternatives, but often the power differential is small. The test corresponding to   $(\gamma,\eta)=(1/2,1/2)$ also performs well against competitors. In fact in the case of the ``$\sin$" alternative, even larger values of the pair $(\gamma,\eta)$ lead to improved power, but at the same time all tests find it rather hard to distinguish an  ``abs" misspecification alternative of the $\pi$ function.

 \begin{table}[htpb]
\centering 
\begin{tabular}{ll|rrrrrrr|rr}
  \hline\hline
  &&\multicolumn{7}{c|}{$\Delta_{T,W}$ with $\rho=0$ and $(\gamma,\eta)$ }&\multicolumn{2}{c}{}\\
 Alt. &$T$ &  $(\frac14,\frac14)$  & $(\frac12,\frac12)$  & $(\frac12,1)$ & $(1,\frac12)$ & (1,1) & (1,2) & (2,2) & $\Delta^{(0)}_{T}$ &$\Delta^{(1)}_{T}$ \\
\hline
\multirow{4}{*}{\rot{ARX(2)}}  &100 & 0.057 & 0.046 & 0.046 & 0.031 & 0.023 & 0.017 & 0.018 & 0.039 & 0.023 \\ 
   & 200 & 0.115 & 0.095 & 0.074 & 0.056 & 0.038 & 0.025 & 0.018 & 0.093 & 0.063 \\ 
   & 500 & 0.291 & 0.274 & 0.229 & 0.171 & 0.128 & 0.063 & 0.054 & 0.268 & 0.243 \\ 
   & 1000 & 0.545 & 0.602 & 0.571 & 0.484 & 0.270 & 0.212 & 0.148 & 0.508 & 0.243 \\ 
   \hline
\multirow{4}{*}{\rot{NB}}   & 100 & 0.893 & 0.898 & 0.872 & 0.894 & 0.838 & 0.723 & 0.696 & 0.873 & 0.795 \\ 
   & 200 & 0.997 & 0.996 & 0.994 & 0.996 & 0.988 & 0.923 & 0.905 & 0.996 & 0.978 \\ 
   & 500 & 1.000 & 1.000 & 1.000 & 1.000 & 1.000 & 1.000 & 0.998 & 1.000 & 1.000 \\ 
   & 1000 & 1.000 & 1.000 & 1.000 & 1.000 & 1.000 & 1.000 & 1.000 & 1.000 & 1.000 \\ 
   \hline
\multirow{4}{*}{\rot{sin}}   & 100 & 0.075 & 0.172 & 0.400 & 0.321 & 0.421 & 0.375 & 0.258 & 0.062 & 0.056 \\ 
   & 200 & 0.140 & 0.444 & 0.730 & 0.663 & 0.750 & 0.696 & 0.554 & 0.098 & 0.093 \\ 
   & 500 & 0.355 & 0.935 & 0.995 & 0.991 & 1.000 & 0.991 & 0.953 & 0.213 & 0.191 \\ 
   & 1000 & 0.741 & 1.000 & 1.000 & 1.000 & 1.000 & 1.000 & 1.000 & 0.388 & 0.191 \\ 
   \hline
\multirow{4}{*}{\rot{abs}} & 100 & 0.055 & 0.062 & 0.056 & 0.058 & 0.054 & 0.052 & 0.044 & 0.054 & 0.050 \\ 
   & 200 & 0.050 & 0.053 & 0.051 & 0.047 & 0.049 & 0.055 & 0.043 & 0.051 & 0.057 \\ 
   & 500 & 0.064 & 0.065 & 0.063 & 0.068 & 0.066 & 0.076 & 0.075 & 0.054 & 0.055 \\ 
   & 1000 & 0.065 & 0.066 & 0.071 & 0.073 & 0.082 & 0.099 & 0.091 & 0.059 & 0.055 \\ 
  \hline\hline
  \end{tabular}\caption{Power of test for GARCHX(1,1) model \eqref{eq:lam2} for data generated from ARX(2) model \eqref{eq:lam3},  
 data from the Negative Binomial distribution with $r=3$ (NB), and data from incorrectly specified $\pi$, where the data are generated using $\pi(x;c)=c(\sin x +1)$ (sin) or $\pi(x;c)=c (2-|x|)I[|x|<2]$ (abs), while we test $\pi(x;c)=c(\cos x+1)$.}\label{tab:GAR-H1}
\end{table}

The power of test for \eqref{eq:lam2} for the considered alternatives is shown in Table~\ref{tab:GAR-H1}. The data were generated from an ARX(2) model \eqref{eq:lam3} with  $\omega=0.2,$ $\alpha_1=0.3$, $\alpha_2=0.6$, and $c=0.5$, from 
a Negative Binomial  GARCHX(1,1) model  with $r=3$ and  a pair of  GARCHX(1,1) models with misspecified $\psi$ function (``sin" and ``abs")  for $\omega=0.2$, $\alpha_1=0.3$, $\beta_1=0.3$, $c=0.5$. 

The results in Table~\ref{tab:GAR-H1} are quantitatively similar with those of  Table \ref{tab:AR-H1} for the NB distributional alternative and the ``abs" misspecification of the $\pi$ function, and qualitatively similar for the other two types of alternatives, but in this case it is harder to identify a ARX(2) alternative while the power against a ``sin" misspecification of the $\pi$ function is somewhat lower for all tests. Comparison--wise, the conclusions are also analogous to those drawn from the results of  Table \ref{tab:AR-H1}, overall favouring the tests based on $\Delta_{T,W}$, with   $\gamma=\eta=1/4$ and $\gamma=\eta=1/2$. 
 It is interesting to see that the procedure leads to reasonable results even if  $\eta<2$. The corresponding bootstrap test keeps the prescribed level, and  this choice seems to be even preferable against some alternatives.

\section{Real data analysis}\label{sec:data}

Road crashes have not only a devastating impact on victims and their families, but they also bear a heavy economic cost, see, e.g., \cite[Table 2.9]{peden}.
% in terms of both  productivity lost and all healthcare resources needed. 
The total cost of motor vehicle collisions involve
lost productivity, medical costs, legal and court costs, emergency services, insurance administration, travel delay, property damage and workplace losses, \citep{blincoe}.  Among all accidents, injuries to
pedestrians and bicyclists caused 7 \% of the total economic costs and 10 \%  of the total societal harm in US in 2010 \citep{blincoe}.

In this application, we 
analyze daily numbers of accidents involving
at least one pedestrian in Prague, the capital of the Czech Republic, in  years 2018--2019. The data are available at the website of the  Ministry of Interior of the Czech Republic, \verb|https://www.mvcr.cz|.
The considered time series has 730 observations and it consists of rather small counts, with minimum~0, maximum 8, and mean 1.70 collisions per day, see also Figure \ref{fig:acc1}. 
 Zero accidents are observed in  22 \% of days. 
The plot of the sample autocorrelation function (ACF) in Figure \ref{fig:acc1} reveals  
an obvious weekly  pattern. The 1-lag autocorrelation is rather small in magnitude, but it should not be ignored, as already shown in several previous analyses of traffic accident counts, see, e.g., \cite{brijs}, \cite{pedeli2011}.

\begin{figure}[pt]
\includegraphics[width=\textwidth]{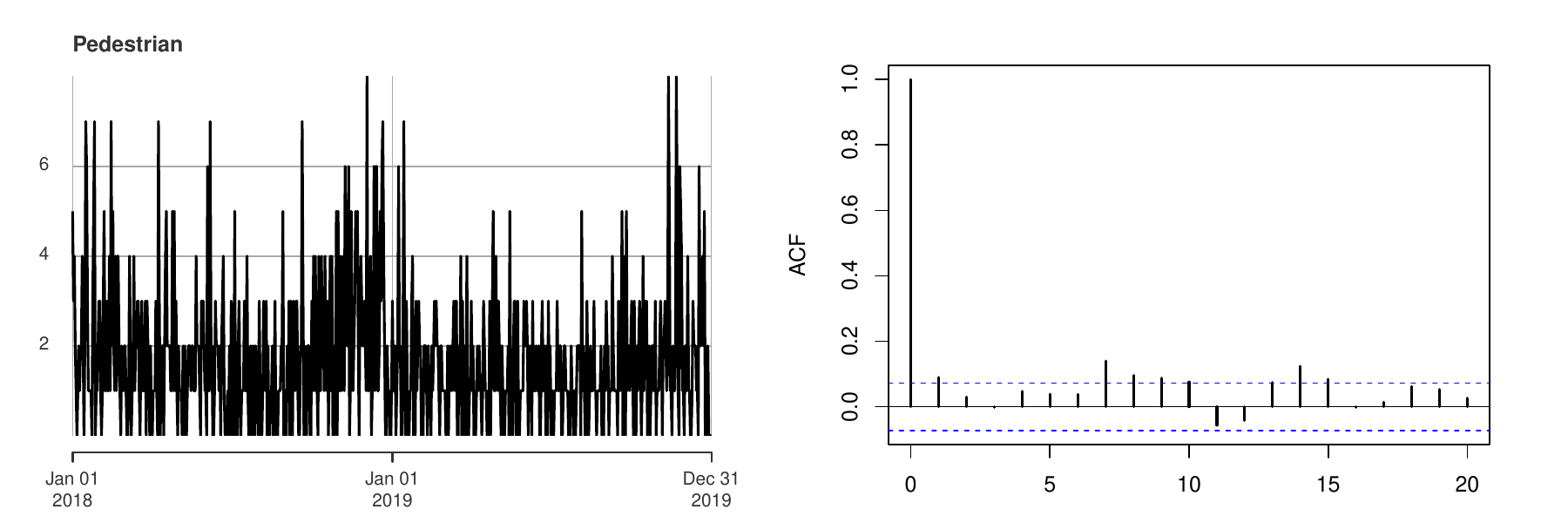}
\caption{Daily numbers of pedestrian collisions (left panel) and the sample ACF (right panel).}\label{fig:acc1}
\end{figure}

It is well known that weather affects both the driver capabilities and road conditions via  visibility impairments, precipitation, and temperature extremes. 
Hence, 
%In order to model the number of pedestrian collisions, not only weakly seasonality was considered, but also a number of meteorological indicators. 
 the following explanatory variables were considered in a model for number of collisions:
 \begin{itemize}\setlength\itemsep{0ex}
\item average daily temperature (AT) in Celsius degrees,
\item daily precipitation (P) in mm,
\item deterministic dummy  covariate indicating a working day (W). %to model weekly seasonality,
\end{itemize} 
The corresponding weather variables were downloaded from the Czech Hydrometeorological Institute,  \verb"https://www.chmi.cz/". 
The AT ranges from $-8.40$ to $31.00$ with mean 12.83. The daily precipitation takes values between 0.00 and 42.30 with mean 1.02, with 67\% of values being 0.

During the model building the average daily temperature was transformed either to a categorical variable with 3 categories (AT in $(5,10]$ , $(10,25]$, otherwise), considered linearly via a variable  
$35-$AT, %on a logarithmic scale via $\log(35-AT)$, 
or different linear dependencies were considered for AT$<5$ and AT$\geq 5$, modelled via variables $\mathrm{I}($AT$<5)\cdot ($AT$+10)$ and  $\mathrm{I}($AT$\geq 5)\cdot (35-$AT).
The transformations (shifts of AT) are performed, because only positive association can be captured by a PARX model and the covariates  are assumed to take positive values. Furthermore, the precipitation was considered either linearly, or via an indicator $\mathrm{I}($P$>0)$. 
Competitive models were fitted by MLE and compared using AIC and BIC, similarly as in \cite{agosto16}. We considered AR(1), AR(2) and GARCH(1,1) structure, but the simplest AR(1) type model of structure
\begin{equation*} %\label{eq:accidents}
\lambda_t = \omega+ \alpha_1 Y_{t-1} + \bb c^\top X_t
\end{equation*}
seems to be the most appropriate, with $\bb X_t$ including the external regressors. 
Note that in this application we model $\lambda_t$ as a function of the current value of $\bb X_t$. The bootstrap goodness-of-fit model was performed for $B=1000$ with various tuning parameters $(\rho,\gamma,\eta)$. Following the results of the simulation study, we set $\rho=0$, $\gamma,\eta\in\{1/4,1/2,1\}$. 
A block bootstrap was applied to the stochastic variables  with block length $l=10$ and $l=30$, while the deterministic variables are held fixed.

The best models, according to the information criteria, seem to be the following two models 
\begin{align*}
M_1: \lambda_t&= 0.440+ 0.036 Y_{t-1}+ 0.245 \mathrm{I}({\rm{P}}_t>0)+0.018(35-{\rm{AT}}_t)+0.981 {\rm{W}}_t, \\
M_2: \lambda_t&= 0.448 +0.034Y_{t-1}+ 0.233 \mathrm{I}({\rm{P}}_t>0) + 0.058(10+{\rm{AT}}_t)\mathrm{I}({\rm{AT}}_t<5)\\
&\quad+ 0.018 (35-{\rm{AT}}_t)\mathrm{I}({\rm{AT}}_t\geq 5)+0.976 {\rm{W}}_t.
\end{align*}
The results of the goodness-of-fit tests for a Poisson model with $\lambda_t$ specified by $M_1$ and $M_2$  are presented in Table~\ref{tab:ped} together with the corresponding information criteria.  The obtained p-values are non-significant for all choices of the tuning parameters $(\gamma,\eta)$ and both considered bootstrap block lengths $l$. AIC slightly prefers $M_2$, while the opposite conclusion is obtained for BIC. As the p-values for the test are generally larger for $M_1$ and due to an easier interpretation, the simpler model $M_1$ is chosen as the final.  The estimated parameters suggest that the difference  in number of pedestrian collisions in working days and during weekends   is in average equal to 1. As expected,  precipitation (and consequently worse visibility) increases a risk of pedestrian collision, but the marginal effect is rather small in magnitude, and the same can be concluded for the marginal effect of temperature. 
 
\begin{table}[ht]
\centering
\begin{tabular}{ll|rrrr|rr}
  \hline   \hline
&& \multicolumn{4}{c|}{Test with $(\gamma,\eta)$}&AIC&BIC\\
& &  $(\frac14,\frac14)$& $(\frac12,\frac12)$ & $(\frac12,1)$ & $(1,\frac12)$ &&\\ 
  \hline
$M_1$  &  $l=10$  & 0.107 & 0.229 & 0.298 & 0.244 &2379.546&2402.511 \\ 
     &  $l=30$  & 0.101 & 0.260 & 0.259 & 0.236 & &\\ 
  \hline
$M_2$  &  $l=10$ &  0.102 & 0.211 & 0.219 & 0.193 &2376.500&2404.058\\ 
   &  $l=30$  & 0.124 & 0.212 & 0.184 & 0.185  & &\\ 
   \hline   \hline
\end{tabular}
 \caption{p-values for the bootstrap goodness-of-fit test for Poisson models $M_1$ and $M_2$ for daily series of pedestrian collisions. }\label{tab:ped}
\end{table}

\section{Discussion and conclusion}\label{sec:concl}

A goodness--of--fit test for count time series with exogenous covariates is proposed. The new test statistic is formulated in terms of Bieren's characterization and may be applied to general models of this type provided that the conditional PGF is specified. The limit distribution of this statistic is obtained under the null hypothesis
 of fixed, but otherwise arbitrary, order and  rather general form of the link function. Asymptotic results for non--null situations are also presented, including consistency of the test against  distributional misspecification as well as against certain types of violations of the hypothesized functional form of the conditional mean. The finite--sample performance of a bootstrap version of the suggested test is investigated by a Monte Carlo study against competing procedures, and corroborates the consistency features of our test. Finally the proposed methodology is used in order to perform explanatory analysis of the main factors affecting road accidents in the city of Prague.     

\subsection*{Acknowledgments}
The research of Marie Hu\v{s}kov\'{a} was partially supported by grant GA\v CR-- 21--13323S. The research of \v{S}\'{a}rka Hudecov\'{a} was supported by grant  GA\v CR 22--01639K.

\bibliographystyle{apalike}

\bibliography{references}

\appendix
\section{Proofs}\label{sec:proofs}
We start with some auxiliary technical lemmas.

\begin{lemma}\label{lem:limit}
Let $\{w_t\}$ be a sequence of random variables such that $\E |w_t|^{\kappa}<\infty$ for some $\kappa>0$ and all $t\geq 1$.  Let $\rho\in(0,1)$. Then
$\rho^t w_t\asto 0$   for $t\to \infty$ and also $T^{-1}\sum_{t=1}^T \rho^t w_t \asto 0$  as $T\to \infty$.
\end{lemma}

\begin{proof}
The first statement follows from the Cantelli lemma, because
\[
\sum_{t=1}^{\infty}\mathsf{P}(|\rho^t w_t| >\eps) \leq \sum_{t=1}^{\infty}\frac{\rho^{t\kappa} \E |w_t|^\kappa}{\eps^{\kappa}}<\infty
\]
by the Markov inequality and the assumption $\E |w_t|^{\kappa}<\infty$. The second claim then follows from the property of the Ces\`{a}ro limit. 
\end{proof}

In the following we often make use of the following simple relations, which hold for $a,b,c,\widetilde{a},\widetilde{b},\widetilde{c}\in \R$:
\begin{align}
ab-\widetilde{a}\widetilde{b} &= a(b-\widetilde{b})+\widetilde{b}(a-\widetilde{a}),\label{eq:ab}\\
abc-\widetilde{a}\widetilde{b}\widetilde{c}& = ab (c-\widetilde{c})+\widetilde{c}(ab-\widetilde{a}\widetilde{b})=ab(c-\widetilde{c})+a\widetilde{c}(b-\widetilde{b})+\widetilde{b}\widetilde{c}(a-\widetilde{a}).\label{eq:abc}
\end{align} 

For functions $r$ and $\pi$ from Assumption~\ref{as:r} recall that $r_{\tth}$, $r_{\tlam}$, $\pi_{\tth}$ stand for vectors of first partial derivatives and define
$
r_{\tth \tth}^{(ij)}(\bb y, \bb l; \tth) = \partial^2 r(\bb y, \bb l; \tth)/\partial \theta_i \partial \theta_j,
$
and let
\begin{equation}\label{eq:der_r}
r_{\tth \tth}(\bb y, \bb l; \tth)= \left(r_{\tth \tth}^{(ij)}(\bb y, \bb l; \tth) \right)_{i,j=1}^K = \frac{\partial^2 r(\bb y, \bb l; \tth)}{\partial \tth \partial \tth^\top}
\end{equation}
be the matrix of second order derivatives of $r$ with respect to $\tth$. 
Similarly, let
\begin{equation}\label{eq:der_r2}
r_{\tth \tlam}(\bb y, \bb \lambda; \tth)=\frac{\partial^2 r(\bb y, \bb \lambda; \tth)}{\partial \tth \partial \tlam^\top}, \  r_{\tlam \tlam}(\bb y, \bb \lambda; \tth)=\frac{\partial^2 r(\bb y, \bb \lambda; \tth)}{\partial \tlam \partial \tlam^\top}, \ 
\pi_{\tth \tth}(\bb x;\tth)=\frac{\partial^2 \pi(\bb x;\tth)}{\partial \tth \partial \tth^\top}
\end{equation}
for $\bb y \in[0,\infty)^p$, $\bb \lambda\in(0,\infty)^q$ and $\tth\in\bb\Theta$.

\begin{lemma} \label{lem:1}
Let \ref{as:h}, \ref{as:stat}, \ref{as:r} hold.
Then  there exists $\kappa>0$ such that
\[
\E \sup_{\tth \in \bb \Theta_0} \lambda_t(\tth) <\infty, \quad \E \sup_{\tth \in \bb \Theta_0} \Bigl\| \frac{\partial \lambda_t (\tth)}{\partial \tth} \Bigr\|^{\kappa} <\infty,
\]
and
\begin{align}
\sup_{\tth \in \bb \Theta_0} \bigl|\widetilde{\lambda}_t(\tth) - \lambda_t(\tth)\bigr| & < V \cdot \rho_1^t, \label{eq:lam.dif}\\
\sup_{\tth \in \bb \Theta_0}  \Bigl\| \frac{\partial \lambda_t (\tth)}{\partial \tth} -  \frac{\partial \widetilde{\lambda}_t (\tth)}{\partial \tth} \Bigr\| &< W_t \cdot \rho_2^t \notag %\label{eq:lam.der.dif},
\end{align}
where $\rho_1,\rho_2\in(0,1)$, and $V$, $W_t$ are random variables such that % is $\mathcal{F}_0$ measurable  random variable 
$\E V^\kappa<\infty$ and 
  $\sup_t \E W_t^{\kappa} <\infty$. 
\end{lemma}

\begin{proof} 
%To see \eqref{eq:lam.dif} for $p=q=1$, write
%\begin{align*}
%|\widetilde{\lambda}(\tth) - \lambda(\tth)| &= |r(Y_{t-1},\widetilde{\lambda}_{t-1}(\tth);\tth) -  r(Y_{t-1},{\lambda}_{t-1}(\tth);\tth)| \leq \beta_1 %|\widetilde{\lambda}_{t-1}(\tth) - {\lambda}_{t-1}(\tth)|\\
%& \leq \beta_1^t(|\widetilde{\lambda}_0| + | {\lambda}_{0}(\tth)|),
%\end{align*}
%where  $0<\beta_1<1$ is from \eqref{eq:Lipsch}. 

The assertions are proved in \cite[Lemma A.1., Lemma A.2, Lemma A.3]{francq21}. 
Using the notation from \eqref{eq:der_r2}, 
 it follows that under our Assumption~\ref{as:r} it holds that $r_{\tlam \tth} =r_{\tth \tlam}= \partial d(\tth)/\partial \tth$, which is 
a continuously differentiable function on  $\tTheta$. Hence the supremum of its norm over the compact set $\tTheta_0$ is finite. In addition,  $r_{\tlam \tlam} = \bb 0$, so assumption A7 of \cite{francq21} is satisfied. 
\end{proof}

\begin{lemma}\label{lem:2}
Let $p=q=1$ and assume that the functions $(r,\pi)$ satisfy Assumptions~\ref{as:h} and \ref{as:r}. Let $\{y_t\}_{t=0}^T$ and $l_0$ be some fixed real numbers and $\{\bb x_t\}_{t=0}^T$ be a sequence of real $m$-dimensional vectors. Also for  $\tth\in\bb \Theta$, 
let $\{l_t (\tth)\}_{t=1}^T$ be recursively defined as
\[
l_t(\tth) = r(y_{t-1},l_{t-1}(\tth);\tth) +\pi(\bb x_{t-1};\tth), \quad t=1,\dots,T
\] 
with $l_0(\tth)=l_0$ for some $l_0\in[0,\infty)$. 
%Denote  
%\[
%r_{\tlam}(y,l;\tth) = \frac{\partial r(y,l;\tth)}{\partial l}, \quad r_{\tth}(y,l;\tth) = \frac{\partial r(y,l;\tth)}{\partial \tth}, \quad \pi_{\tth}(\bb x;\tth) = \frac{\partial \pi(\bb x;\tth)}{\partial \tth}.
%\]
Then using the notation from Assumption~\ref{as:r},
\begin{align*}
\frac{\partial l_t(\tth)}{\partial \tth} &= r_{\tth}(y_{t-1},l_{t-1}(\tth);\tth)+\pi_{\tth}(\bb x_{t-1};\tth) \\
&+ \sum_{k=1}^{t-1} \prod_{j=1}^k r_{\tlam}(y_{t-j},l_{t-j}(\tth);\tth) [r_{\tth}(y_{t-k-1},l_{t-k-1}(\tth);\tth)+\pi_{\tth}(\bb x_{t-k-1};\tth)]. 
\end{align*}
\end{lemma}
\begin{proof}
We have
\[
\frac{\partial l_t(\tth)}{\partial \tth}  = r_{\tlam}(y_{t-1},l_{t-1}(\tth);\tth) \frac{\partial l_{t-1}(\tth)}{\partial \tth} +  r_{\tth}(y_{t-1},l_{t-1}(\tth);\tth)+\pi_{\tth}(\bb x_{t-1};\tth)
\]
for $t\geq 2$ and
\[
\frac{\partial l_1(\tth)}{\partial \tth}  =  r_{\tth}(y_{0},l_{0};\tth)+\pi_{\tth}(\bb x_{0};\tth).
\]
Denote 
\[
A_k = r_{\tlam}(y_{k},l_{k}(\tth);\tth), \quad k\geq 1,  \quad D_k =   r_{\tth}(y_{k},l_{k}(\tth);\tth)+\pi_{\tth}(\bb x_{k};\tth), \quad k\geq 0. 
\]
Then
\begin{align*}
\frac{\partial l_t(\tth)}{\partial \tth}  &= A_{t-1}  \frac{\partial l_{t-1}(\tth)}{\partial \tth} +  D_{t-1} = A_{t-1} \left[A_{t-2} \frac{\partial l_{t-2}(\tth)}{\partial \tth} +D_{t-2}\right]  D_{t-1} \\
&=
 A_{t-1} A_{t-2} A_{t-3} \frac{\partial l_{t-3}(\tth)}{\partial \tth} + A_{t-1}A_{t-2}D_{t-3}+ A_{t-1}D_{t-2}+  D_{t-1} =\dots \\
 &= \sum_{k=1}^{t-1} \prod_{j=1}^k A_{t-j} D_{t-k-1}+D_{t-1}, 
\end{align*}
which is the desired formula.
\end{proof}

\begin{lemma}\label{lem:3}
Let \ref{as:h},  \ref{as:stat}, \ref{as:r} hold. %Let $p=q=1$ and $r_{\tth}(y,l,\tth)$ be Lipchitz in $\tth$ in a sense that
%\[
%\| 
%r_{\tth}(y,l_1,\tth) - r_{\tth}(y,l_2,\tth)\| \leq K(\tth)|l_1-l_2| \quad \text{ for all } l_1,l_2\geq 0 \text{ and}\quad \sup_{\tth \in \Theta_0} K(\tth)<\infty.
%\]
Then there exists $\kappa>0$ such that
\[
\E \sup_{\tth \in \bb \Theta_0}\left\|\frac{\partial \widetilde{\lambda}_t(\tth)}{\partial \tth} \right\|^{\kappa} <\infty
\]
for all $t\geq 1$.
\end{lemma}

\begin{proof}
We present a  proof for $p=q=1$; general orders $(p,q)$ can be treated analogously. 
Under the assumptions in \ref{as:r},  $r_{\tlam}(Y_{t-1},\widetilde{\lambda}_{t-1}(\tth);\tth)=d(\tth)$ is a non-random scalar such that $\sup_{\tth\in \bb\Theta_0} |d(\tth)|=\gamma<1$. 
It then follows from Lemma~\ref{lem:2} that 
\begin{align*}
\sup_{\tth\in\bb\Theta_0} \left\|\frac{\partial \widetilde{\lambda}_t(\tth)}{\partial \tth} \right\|&\leq \sup_{\tth\in\bb\Theta_0} |r_{\tth}(Y_{t-1},\lambda_{t-1}(\tth);\tth)|
+  \sup_{\tth\in\bb \Theta_0} |r_{\tth}(Y_{t-1},\widetilde{\lambda}_{t-1}(\tth);\tth) -r_{\tth}(Y_{t-1},{\lambda}_{t-1}(\tth);\tth)| \\
&+\sup_{\tth\in\bb \Theta_0} |\pi_{\tth}(\bb X_{t-1};\tth)| + \sum_{k=1}^{t-1} \gamma^k |r_{\tth}(Y_{t-k-1},{\lambda}_{t-k-1}(\tth);\tth)|\\
&+ \sum_{k=1}^{t-1} \gamma^k |r_{\tth}(Y_{t-k-1},\widetilde{\lambda}_{t-k-1}(\tth);\tth)-r_{\tth}(Y_{t-k-1},{\lambda}_{t-k-1}(\tth);\tth)|\\
&+ \sum_{k=1}^{t-1} \gamma^k \sup_{\tth\in\bb \Theta_0}|\pi_{\tth}(\bb X_{t-k-1};\tth)|. 
\end{align*}
 If $r_{\tlam \tth}$, $r_{\tlam \tth}$, $r_{\tth \tth}$ stand for the second partial derivatives of $r$ defined in \eqref{eq:der_r} and \eqref{eq:der_r2}, then  it follows from Assumption \ref{as:r} that  $r_{\tlam \tth}(y,\lambda,\tth)= \partial d(\tth) / \partial \tth$
 for all $y\in[0,\infty)$ and $\lambda\in(0,\infty)$. By the mean value theorem
 \[
 \bigl|r_{\tth}(y,\widetilde{\lambda}_{t-1}(\tth);\tth) -r_{\tth}(y,{\lambda}_{t-1}(\tth);\tth)\bigr| = \left|\frac{\partial d(\tth)}{\partial \tth} \right| |\widetilde{\lambda}_{t-1}(\tth)-{\lambda}_{t-1}(\tth)|
 \]
 for any $y\in[0,\infty)$. 
 Since 
 $\partial d(\tth) / \partial \tth$
 
 is continuous on the compact set $\bb \Theta_0$,  the supremum of its norm over $\bb \Theta_0$ is finite and non-random. 
The assertion then follows from %the mean value theorem,  
\eqref{eq:Esup} and Lemma~\ref{lem:1}. 
\end{proof}

\begin{lemma}\label{lem:4}
Let \ref{as:h}, \ref{as:stat}, \ref{as:r} hold. Then there exist $\beta_1,\beta_2\in(0,1)$ such that
\[
\sup_{\tth \in \bb \Theta_0} \left| \frac{\partial \lambda_t(\tth)}{\partial \theta_i} \frac{\partial \lambda_t(\tth)}{\partial \theta_j}  - \frac{\partial \widetilde\lambda_t(\tth)}{\partial \theta_i} \frac{\partial \widetilde\lambda_t(\tth)}{\partial \theta_j} 
\right|< V_{t1} \cdot \beta_1^t, 
\]
for all $i,j=1,\dots,K$, and
\[
\sup_{\tth \in \bb \Theta_0}  \Bigl\| \frac{\partial^2 \lambda_t (\tth)}{\partial \tth \partial \tth^\top} -  \frac{\partial^2 \widetilde{\lambda}_t (\tth)}{\partial \tth\partial \tth^\top} \Bigr\| <  V_{t2} \cdot \beta_2^t %\label{eq:lam.der.dif2},
\]
where $V_{ti}$, $i=1,2$, are positive random variables such that $\E V_{ti}^\kappa<\infty$ for some $\kappa>0$.  
\end{lemma}
\begin{proof} 
 It follows from \eqref{eq:ab} that
\begin{align*}
\frac{\partial \widetilde{\lambda}_t(\tth)}{\partial \theta_j}\frac{\partial \widetilde{\lambda}_t(\tth)}{\partial \theta_i }&=
 \frac{\partial {\lambda}_t(\tth)}{\partial \theta_j}\frac{\partial {\lambda}_t(\tth)}{\partial \theta_i } + \frac{\partial \widetilde{\lambda}_t(\tth)}{\partial \theta_j}\left[ \frac{\partial \widetilde{\lambda}_t(\tth)}{\partial \theta_i } -\frac{\partial {\lambda}_t(\tth)}{\partial \theta_i } \right] \\
& + \frac{\partial {\lambda}_t(\tth)}{\partial \theta_i }  \left[ \frac{\partial \widetilde{\lambda}_t(\tth)}{\partial \theta_j } -\frac{\partial {\lambda}_t(\tth)}{\partial \theta_j} \right],
\end{align*}
so by Lemma~\ref{lem:1} 
\[
\sup_{\tth \in \bb \Theta_0}  \left| \frac{\partial \lambda_t(\tth)}{\partial \theta_i} \frac{\partial \lambda_t(\tth)}{\partial \theta_j}  - \frac{\partial \widetilde\lambda_t(\tth)}{\partial \theta_i} \frac{\partial \widetilde\lambda_t(\tth)}{\partial \theta_j} 
\right| \leq \sup_{\tth \in \bb \Theta_0} \left|  \frac{\partial \widetilde{\lambda}_t(\tth)}{\partial \theta_j} \right| W_t \rho_2^t + \sup_{\tth \in \bb\Theta_0} \left|  \frac{\partial {\lambda}_t(\tth)}{\partial \theta_j} \right| W_t \rho_2^t,
\]
and the statement follows from Lemmas~\ref{lem:1} and \ref{lem:3}.

The second part of the assertion is proved solely for $p=q=1$, while for general orders the proof proceeds in an analogous way. %
 Recall the notation \eqref{eq:der_r} and \eqref{eq:der_r2}. 
If $l_t(\tth)$ is defined as in Lemma~\ref{lem:2}, then
% Using the notation from the proof of Lemma~\ref{lem:3} for the second order partial derivatives of $r$, we get that  
\begin{align*}
\frac{\partial l_{t}(\tth)}{\partial \tth \partial \tth^\top} &= r_{\tth \tth}(y_t,  l_{t-1}(\tth);\tth) + r_{\tth \tlam}(y_t,  l_{t-1}(\tth);\tth) \frac{\partial l_{t-1}(\tth)}{\partial \tth^\top} \\
&+  r_{\tlam \tlam}(y_t,  l_{t-1}(\tth);\tth) \frac{\partial l_{t-1}(\tth)}{\partial \tth}  \frac{\partial l_{t-1}(\tth)}{\partial \tth^\top} + r_{\tlam}(y_t,  l_{t-1}(\tth);\tth)\frac{\partial l_{t-1}(\tth)}{\partial \tth \partial \tth^\top}+ \pi_{\tth \tth}(\bb x_{t-1};\tth).
\end{align*}
Due to Assumptions \ref{as:r},  $r_{\tlam}(y,  \lambda;\tth)=d(\tth)$ for any $y,\lambda\in[0,\infty)$, so $r_{\tlam \tlam}=0$ and $r_{\tlam \tth}(y,l;\tth) = \partial d(\tth)/\partial \tth$. Denote as  
$d_{\tth}^{(j)}(\tth)=\partial d(\tth)/\partial \theta_j$  $j\in\{1,\dots,K\}$ and
$d_{\tth}(\tth)=\partial d(\tth)/\partial \tth=\Bigl(d_{\tth}^{(j)}(\tth)\Bigr)_{j=1}^K$. 
Then 
%d_{\tth}$ is continuous in $\tth$ and
\begin{align*}
\Bigl| \frac{\partial^2 \lambda_t(\tth)}{\partial \theta_i\partial \theta_j} - \frac{\partial^2 \widetilde{\lambda}_t(\tth)}{\partial \theta_i\partial \theta_j}\Bigr| &\leq
|r_{\tth \tth}^{(ij)}(Y_{t-1}, \lambda_{t-1}(\tth);\tth) - r_{\tth \tth}^{(ij)}(Y_{t-1}, \widetilde{\lambda}_{t-1}(\tth);\tth)|\\
&+
\Bigl| d_{\tth}^{(j)}(\tth) \Bigr| \left| \frac{\partial \lambda_t(\tth)}{\partial \theta_i} -\frac{\partial \widetilde{\lambda}_t(\tth)}{\partial \theta_i}  \right| 
%\Bigl|r_{\tth \tlam}^{(i)}(Y_{t-1}, \lambda_{t-1}(\tth);\tth) \Bigr| \left| \frac{\partial \lambda_t(\tth)}{\partial \theta_j} -\frac{\partial \widetilde{\lambda}_t(\tth)}{\partial \theta_j}  \right|\\
%&+ \left| \frac{\partial \widetilde{\lambda}_t(\tth)}{\partial \theta_j} \right| \Bigl|d_{\tth}^{(i)}(Y_{t-1}, {\lambda}_{t-1}(\tth);\tth) -  d_{\tth \tlam}^{(i)}(Y_{t-1}, \widetilde{\lambda}_{t-1}(\tth);\tth)\Bigr|\\
%&+ \left| \frac{\partial \widetilde{\lambda}_t(\tth)}{\partial \theta_i} \right| \Bigl|r_{\tth \tlam}^{(j)}(Y_{t-1}, {\lambda}_{t-1}(\tth);\tth) -  r_{\tth \tlam}^{(j)}(Y_{t-1}, \widetilde{\lambda}_{t-1}(\tth);\tth)\Bigr|
%& + \left| \frac{\partial^2 \lambda_t(\tth)}{\partial \theta_i\partial \theta_j} \right| \left|r_{\tlam}(Y_{t-1}, \lambda_{t-1}(\tth);\tth)-  r_{\tlam}(Y_{t-1}, \widetilde\lambda_{t-1}(\tth);\tth)\right|\\
+  |d(\tth)|   \left| \frac{\partial^2 \lambda_t(\tth)}{\partial \theta_i\partial \theta_j}  - \frac{\partial^2 \widetilde\lambda_t(\tth)}{\partial \theta_i \partial \theta_j}\right|. 
\end{align*}
Furthermore, for $i,j\in\{1,\dots,K\}$
\[
\frac{\partial^3 r(y,\lambda;\tth)}{\partial \lambda \partial \theta_i \partial \theta_j} = \frac{\partial^2 }{\partial \theta_i \partial \theta_j}\left(\frac{\partial r(y,\lambda;\tth)}{\partial \lambda}\right)=\frac{\partial^2 d(\tth)}{\partial \theta_i \partial \theta_j}=:d_{\tth \tth}^{(ij)}(\tth),  
\]
so the mean value theorem implies that
\[
|r_{\tth \tth}^{(ij)}(Y_{t-1}, \lambda_{t-1}(\tth);\tth) - r_{\tth \tth}^{(ij)}(Y_{t-1}, \widetilde{\lambda}_{t-1}(\tth);\tth)| = |d_{\tth\tth}^{(ij)}(\tth)|\cdot |\lambda_{t-1}(\tth) - \widetilde{\lambda}_{t-1}(\tth)|.
\]
Assumption \ref{as:r} ensures that the function $d_{\tth \tth}^{(ij)}$ is 
%$r_{\tlam \tth \tth} =r_{ \tth \tth \tlam} =\partial^2 d(\tth)/\partial \tth \partial \tth^\top$ is 
continuous on $\bb \Theta_0$,
so it follows from Lemma~\ref{lem:1} that
\[
\sup_{\tth \in \bb \Theta_0} \Bigl| \frac{\partial^2 \lambda_t(\tth)}{\partial \theta_i\partial \theta_j} - \frac{\partial^2 \widetilde{\lambda}_t(\tth)}{\partial \theta_i\partial \theta_j}\Bigr| 
\leq K_1 \cdot V \cdot  \rho_1^t +  K_2 \cdot W_t \rho_2^t + |d(\tth)|   \left| \frac{\partial^2 \lambda_t(\tth)}{\partial \theta_i\partial \theta_j}  - \frac{\partial^2 \widetilde\lambda_t(\tth)}{\partial \theta_i \partial \theta_j}\right| 
\]
and the assertion then easily follows from $|d(\tth)|<1$. 
\end{proof}

\noindent{\it Proof of Theorem~\ref{th1}.}

In the following we write $\widehat{\tth}$ instead of $\widehat{\tth}_T$, and $g(u,\lambda)$ for the PGF $g_{\lambda}(u)$ of the distribution $\mathcal{F}_{\lambda}$. We consider for simplicity  the case $p=q=1$, 
so $\bb Z_t=(Y_{t-1},\lambda_{t-1},\tX_{t-1}^\top)^\top$, and
\[
h(\bb Z_t;\tth) = r(Y_{t-1},\lambda_{t-1}(\tth);\tth)+\pi(\tX_{t-1};\tth),
\]
with $d=m+2$, where $m$ is the dimension of $\tX_t$. The more general situation can be treated analogously.

%Let $\tth_0$ be the true value of the parameter, and

Let $\lambda_t(\tth)$ be  the  stationary solution to \eqref{rate2} for parameter $\tth$. Furthermore, let $\widetilde{\lambda}_t(\tth)$ be defined recursively for some initial values for $Y_0$ and $\widetilde{\lambda}_0$ as
\[
%\widetilde{\lambda}_t(\tth) = r(Y_{t-1},\dots,Y_{t-p},\widetilde{\lambda}_{t-1}(\tth),\dots, \widetilde{\lambda}_{t-q}(\tth);\tth) +\pi(\tX_{t-1};\tth).
\widetilde{\lambda}_t(\tth) = r(Y_{t-1},\widetilde{\lambda}_{t-1}(\tth);\tth) +\pi(\tX_{t-1};\tth).
\]
Similarly, let $\widetilde{\bb Z}_t(\tth)$ be  defined as
\[
 %\widetilde{\bb Z}_t(\tth)=(Y_{t-1},\dots,Y_{t-p},  \widetilde{\lambda}_{t-1}(\tth),\dots,\widetilde{\lambda}_{t-q}(\tth),\boldsymbol X^{\top}_{t-1})^\top,
 \widetilde{\bb Z}_t(\tth)=(Y_{t-1},  \widetilde{\lambda}_{t-1}(\tth),\boldsymbol X^{\top}_{t-1})^\top,
\]
where we now stress the dependence on $\tth$. 
Then 
\[
\widehat{\lambda}_t = \widetilde{\lambda}_t(\widehat{\tth}), \quad \widehat{\bb Z}_t =  \widetilde{\bb Z}_t(\widehat{\tth}).
\]

Recall the notation
$
\alpha(\bb v_1,\bb v_2)=\cos\left(\bb v_1^\top \bb v_2\right)+\sin\left(\bb v_1^\top \bb v_2 \right),
$
and define 
\begin{align}\label{eq:alpha_tilde}
\alpha_t(\tth,\bb v) &= \alpha\left( \bb v, \tZ_t(\tth) \right) = \cos(\bb v^\top{\tZ}_t(\tth))+\sin(\bb v^\top{\tZ}_t(\tth)),  \\
\widetilde{\alpha}_t(\tth,\bb v) &= \alpha\left( \bb v, \widetilde{\tZ}_t(\tth) \right) = \cos(\bb v^\top\widetilde{\tZ}_t(\tth))+\sin(\bb v^\top\widetilde{\tZ}_t(\tth)),
%&= \cos(v_1+v_2 Y_{t-1} +v_3 \widetilde{\lambda}_{t-1}(\tth) + v_4 X_{t-1})+\sin(v_1+v_2 Y_{t-1} +v_3 \widetilde{\lambda}_{t-1}(\tth) + v_4 X_{t-1}).
\end{align}
where
\begin{equation}\label{eq:vZ}
\bb v^\top\widetilde{\tZ}_t(\tth) = v_1 Y_{t-1}+v_2 \widetilde{\lambda}_{t-1}(\tth)+v_3 X_{t-1,1}+\dots + v_{m+2} X_{t-1,m}. 
\end{equation}

Also notice that assumption \ref{as:W} implies $W(u,\bb v) = w(u)\omega(\bb v)$, with
  \[
  \int_{\R^d}\sin(\bb v^{\top}\bb x) \omega (\bb v)  {\rm{d}}\bb v =0 %, \quad   \int_{\mathbb{R}^d} \cos(\bb v^{\top}\bb x)\sin(\bb v^{\top}\bb x) \omega (\bb v)  {\rm{d}}\bb v =0,
  \]
   for all $ \boldsymbol x \in \mathbb{R}^d$.
Hence, we can get an equivalent expression for
$\Delta_{T,W}$ as
\[
\Delta_{T,W} = \int_{0}^1\int_{\R^d}\left| \frac{1}{\sqrt{T}}\sum_{t=1}^T \widehat{\eps}_t(u) \widetilde{\alpha}_t(\widehat{\tth},\bb v) \right|^2 w(u)\omega(\bb v) {\rm{d}}\bb v {\rm{d}}u .
\]
Thus, it  suffices to investigate the
asymptotic  behavior of
\begin{equation*} %\label{delta-p1}
\widetilde{\Delta}_{T,W} = \int_{\R^d} \int_0^1 |\widetilde{B}_T(u,\bb v,\widehat{\tth})|^2 w(u) \omega (\bb v) {\rm{d}}u {\rm{d}}\bb v,
\end{equation*}
where  
\begin{equation} \label{delta-p2}
\widetilde{B}_T(u,\bb v,\tth)=\frac{1}{\sqrt{T}}\sum_{t=2}^T \left[u^{Y_t} - g(u,\widetilde{\lambda}_t(\tth))\right]\widetilde{\alpha}_t(\tth,\bb v).
\end{equation}
We will first use the  Taylor expansion for $\widetilde{B}_T(u,\bb v, \widehat{\tth})$ in $\tth_0$ and then we will approximate $\widetilde{\lambda}_t(\tth_0)$ by the stationary variables $\lambda_t(\tth_0)$. This will  allow us to express $\widetilde{B}_T(u,\bb v, \widehat{\tth})$ as a  sum of a function of stationary variables $(Y_t, \lambda_t(\tth_0))$ and a negligible remainder.

The  Taylor expansion %for $\widetilde{B}_T(u,\bb v, \widehat{\tth})$ in $\tth_0$ 
gives % {\bf In the below only first line with number}
\begin{align*} %\label{eq:taylor}
\widetilde{B}_T(u,\bb v,\widehat{\tth})&=\underbrace{\widetilde{B}_T(u,\bb v, \tth_0)}_{A_1}+\underbrace{\Big(\frac{1}{\sqrt{T}}\frac{\partial \widetilde{B}_T(u,\bb v,\tth_0)}{\partial \tth} \Big)^\top\sqrt{T}(\widehat{\tth}_T-\tth_0)}_{A_2}\\
\notag &+ \underbrace{\frac12 \sqrt{T} (\widehat{\tth}_T-\tth_0)^\top\left[\frac{1}T \mathbb{B}(\tth^*) \right]\sqrt{T}(\widehat{\tth}_T-\tth_0)}_{A_3},
\end{align*}
where
$\mathbb{B}(\tth^*)$ is a matrix of the second order derivates $\left(\frac{\partial^2 \widetilde{B}_T(u,\bb v,\tth)}{\partial \theta_i\theta_j} \right)_{i,j=1}^K $ evaluated at $\tth=\tth^*$,
where $\tth^*$ is a point between $\widehat{\tth} $  and $\tth_0$.   We will  show that the first two terms $A_1$ and $A_2$ are influential, while   the term $A_3$ is asymptotically negligible.

\medskip

\noindent{\it Treating $A_1$:}
Introduce 
\[
B_T(u,\bb v,\tth)=\frac{1}{\sqrt{T}}\sum_{t=2}^T  \left[u^{Y_t} - g(u,{\lambda}_t(\tth))\right]\alpha_t(\tth,\bb v).
\]
Then from \eqref{eq:ab}
\begin{align*}
\widetilde{B}_T(u,\bb v,\tth) - B_T(u,\bb v,\tth)
=&  \frac{1}{\sqrt{T}}\sum_{t=2}^T \left[u^{Y_t} - g(u,\lambda_t(\tth))\right][\widetilde{\alpha}_t(\tth,\bb v) - {\alpha}_t(\tth,\bb v)] \\
&+\frac{1}{\sqrt{T}}\sum_{t=2}^T  \widetilde{\alpha}_t(\tth,\bb v)[ g(u,\widetilde{\lambda}_t(\tth))- g(u,\lambda_t(\tth))].
\end{align*}
It follows from \eqref{eq:alpha_tilde}, \eqref{eq:vZ}, \eqref{eq:lam.dif} and \ref{as:g} that
\begin{align}
|\alpha_t(\tth,\bb v)  - \widetilde{\alpha}_t(\tth,\bb v)|& \leq 2 |v_2| |\lambda_t (\tth) -\widetilde{\lambda}_t (\tth) | \leq 2|v_2| V \rho_1^t,\label{eq:alpha.dif}\\
| g(u,\widetilde{\lambda}_t(\tth))- g(u,\lambda_t(\tth))|&\leq M |\lambda_t (\tth) -\widetilde{\lambda}_t (\tth) |  \leq M V \rho_1 ^t \notag
\end{align}
for some random variable $V$ and $\rho_1\in(0,1)$. Also recall that $g(u,\lambda)$ is defined as the PGF of $Z\sim\mathcal{F}_{\lambda}$, i.e. $g(u,\lambda)=\E u^Z$, and consequently $|g(u,\lambda)|\leq 1$ for all $u\in[0,1]$. 
%%and also $|\alpha_t(.)|\leq 2$ and $|e^{\widetilde{\lambda}_t(u-1)}|\leq 1$. Hence,
%\[
%\widetilde{B}_T(u,v,\tth) - B_T(u,v,\tth)
%= \leq \frac{1}{\sqrt{T}}\sum_{t=2}^T \left[u^{Y_t} - g(u;\lambda_t(\tth))\right]|(\alpha(\tth,\bb v) - \widetilde{\alpha}(\tth,\bb v)| +\frac{1}{\sqrt{T}}\sum_{t=2}^T  \widetilde{\alpha}(\tth,\bb v)[ g(u;\widetilde{\lambda}_t(\tth))- g(u;\lambda_t(\tth))],
%\]
Hence,  it follows from Lemma~\ref{lem:limit} that
\begin{align*}
|\widetilde{B}_T(u,\bb v,\tth) &- B_T(u,\bb v,\tth)|
%\leq \frac{2}{\sqrt{T}}\sum_{t=2}^T |\alpha(\tth,\bb v) - \widetilde{\alpha}(\tth,\bb v)| +   \frac{2}{\sqrt{T}}\sum_{t=2}^T | g(u;\widetilde{\lambda}_t(\tth))- g(u;\lambda_t(\tth))|\\
\leq \frac{2}{\sqrt{T}} \sum_{t=2}^T 4|v_2| V \rho_1^t + \frac{2}{\sqrt{T}} \sum_{t=2}^TM V \rho_1^t    \asto 0
\end{align*}
as $T\to \infty$. % uniformly in $u$ and $\bb v$.

\medskip
\noindent{\it Treating $A_2$:} We have
\[
\frac{1}{\sqrt{T}}\frac{\partial \widetilde{B}_T}{\partial \tth}(u,\bb v,\tth) =  % \frac{1}{T}\sum_{t=2}^T \ell_t(u,\bb v, \tth)   = 
 \frac{1}{T}\sum_{t=2}^T \Bigl[\ell_{1,t}(u,\bb v, \tth)+\ell_{2,t}(u,\bb v, \tth)\Bigr],
\]
where
\begin{align*}
\ell_{1,t}(u,\bb v,\tth) =& - g_z(u,\widetilde{\lambda}_t(\tth)) \frac{\partial \widetilde{\lambda}_t(\tth)}{\partial \tth} \widetilde{\alpha}_t(\tth,\bb v)  \\
\ell_{2,t}(u,\bb v,\tth) =&  \left[u^{Y_t} - g(u,\widetilde{\lambda}_t(\tth))\right] \frac{\partial \widetilde{\alpha}_t(\tth,\bb v)}{\partial \tth},
\end{align*}
with $g_z(u,z) = \frac{\partial g(u,z)}{\partial z}$.

Similarly let \begin{align*}
\ell_{1,t}^0(u,\bb v,\tth) =& - g_z(u,{\lambda}_t(\tth)) \frac{\partial  {\lambda}_t(\tth)}{\partial \tth}  {\alpha}_t(\tth,\bb v)  \\
\ell_{2,t}^0(u,\bb v,\tth) =&  \left[u^{Y_t} - g(u,{\lambda}_t(\tth))\right] \frac{\partial \alpha_t(\tth,\bb v)}{\partial \tth}.
\end{align*}
Then by \eqref{eq:abc} 
\begin{align*}
\ell_{1,t}(u,\bb v,\tth) -& \ell_{1,t}^0(u,\bb v,\tth)  = g_z(u,\widetilde{\lambda}_t(\tth))  \widetilde{\alpha}_t(\tth,\bb v) \left[ \frac{\partial \widetilde{\lambda}_t(\tth)}{\partial \tth} -\frac{\partial {\lambda}_t(\tth)}{\partial \tth}\right]\\
&+ \frac{\partial {\lambda}_t(\tth)}{\partial \tth} g_z(u,\widetilde{\lambda}_t(\tth)) [  {\alpha}_t(\tth,\bb v) - \widetilde{\alpha}_t(\tth,\bb v)] \\
&+ \frac{\partial {\lambda}_t(\tth)}{\partial \tth} {\alpha}_t(\tth,\bb v)  [ g_z(u,{\lambda}_t(\tth)) - g_z(u,\widetilde{\lambda}_t(\tth))].
\end{align*}
Now, in view of Assumption~\ref{as:g}, Lemma~\ref{lem:1} and \eqref{eq:alpha.dif} we have, 
\begin{align*}
\bigl|\ell_{1,t}(u,\bb v,\tth) -& \ell_{1,t}^0(u,\bb v,\tth) \bigr| \leq  K_1 W_t \rho_2^t + K_2 \left\|\frac{\partial {\lambda}_t(\tth)}{\partial \tth}\right\|  |v_2| V \rho_1^t+K_3 \left\| \frac{\partial {\lambda}_t(\tth)}{\partial \tth}\right\| V \rho_1^t
\end{align*}
for some finite constants $K_i>0$, $i=1,2,3$. Therefore, it follows from Lemma~\ref{lem:limit} that
\[
\frac{1}{T} \sum_{t=1}^T \ell_{1,t}(u,\bb v,\tth) - \frac{1}{T} \sum_{t=1}^T \ell_{1,t}^0(u,\bb v,\tth) \asto 0,
\]
as $T\to \infty$. %, and  the convergence is uniform in $u\in[0,1]$ and $\bb v\in\R^d$. 
Assumption \ref{as:stat} then implies that 
\[
\frac{1}{T} \sum_{t=1}^T \ell_{1,t}(u,\bb v,\tth)   \asto \E \ell_{1,t}^0(u,\bb v,\tth) =\bb \beta(u,v),
\]
where $\bb \beta(\cdot,\cdot)$ is defined before the statement of Theorem~\ref{th1}. 
Furthermore, observe that
\[
\frac{\partial \widetilde{\alpha}_t (\tth,\bb v)}{\partial \tth} = v_2\frac{\partial \widetilde{\lambda}_{t-1}(\tth)}{\partial \tth}\left[ \cos(\bb v^\top\widetilde{\tZ}_t(\tth)) -\sin(\bb v^\top\widetilde{\tZ}_t(\tth)) \right],
\]
and thus,  
\begin{align*}
\left\| \frac{\partial \widetilde{\alpha}_t (\tth,\bb v)}{\partial \tth} \right\|& \leq 2 |v_2| \left\| \frac{\partial \widetilde{\lambda}_{t-1} (\tth,\bb v)}{\partial \tth} \right\|,\\
 \left\| \frac{\partial \widetilde{\alpha}_t (\tth,\bb v)}{\partial \tth}  -  \frac{\partial {\alpha}_t (\tth,\bb v)}{\partial \tth} \right\| &\leq K_4 |v_2| \left\| \frac{\partial \widetilde{\lambda}_{t-1} (\tth,\bb v)}{\partial \tth}  -  \frac{\partial {\lambda}_{t-1} (\tth,\bb v)}{\partial \tth} \right\|\leq K_4 |v_2| W_t \rho_2^{t-1},
\end{align*}
for some $K_4<\infty$. Hence, the sequence $\{\ell_{2,t}\}$ can be treated by the same arguments and it follows that 
\[
\frac{1}{T} \sum_{t=1}^T \ell_{1,t}(u,\bb v,\tth) \asto \E  \ell_{2,t}^0(u,\bb v,\tth) =\bb 0.
\]
Therefore, the term $A_2$ behaves asymptotically as $ \sqrt{T}(\widehat{\tth}-\tth_0)^\top \bb\beta(u,\bb v)$, so in view of Assumption \ref{as:est} we have,
\[
A_2= \frac{1}{\sqrt{T}}\sum_{t=2}^T \bb s_t(\tth_0,Y_{t},\mathcal{I}_{t-1})^\top \bb \beta(u,\bb v) +o_{\mathsf P}(1). 
\]

\medskip
\noindent{\it Treating $A_3$:} 
The second order derivatives are
\[
\frac{\partial \widetilde{B}_T(u,\bb v, \tth^*)}{\partial \theta_i \theta_j} = 
\frac{1}{\sqrt{T}} \sum_{t=2}^T \xi_{t,ij}(u,\bb v, \tth^*),
\]
where
\begin{align*}
\xi_{t,ij}(u,\bb v, \tth)&= -g_{zz}(u,\widetilde{\lambda}_t(\tth)) \frac{\partial \widetilde{\lambda}_t(\tth)}{\partial \theta_j}\frac{\partial \widetilde{\lambda}_t(\tth)}{\partial \theta_i} \widetilde{\alpha}_t(\tth,\bb v)
- g_z(u,\widetilde{\lambda}_t(\tth)) \frac{\partial^2 \widetilde{\lambda}_t(\tth)}{\partial \theta_j\partial \theta_i}\widetilde{\alpha}_t(\tth,\bb v) \\
&- g_z(u,\widetilde{\lambda}_t(\tth))\frac{\partial \widetilde{\lambda}_t(\tth)}{\partial \theta_i}
\frac{\partial \widetilde{\alpha}_t(\tth,\bb v)}{\partial \theta_j} +\Bigl[u^{Y_t} - g(u,\widetilde{\lambda}_t(\tth)) \Bigr] \frac{\partial^2 \widetilde{\alpha}_t(\tth,\bb v)}{\partial \theta_j\partial \theta_i} \\
&- g_z(u,\widetilde{\lambda}_t(\tth)) \frac{\partial \widetilde{\lambda}_t(\tth)}{\partial \theta_j}\frac{\partial \widetilde{\alpha}_t(\tth)}{\partial \theta_i},
\end{align*}
 with $g_{zz}(u,z)=\partial^2 g(u,z)/\partial z^2$.
To prove that $T^{-1}\mathbb{B}(\tth^*)$ converges in probability to a zero matrix,  it suffices to show that the probability limit of 
$T^{-1}  \sum_{t=2}^T \xi_{t,ij}(u,\bb v, \tth^*)$ as $T\to\infty$ is finite, for all $i,j=1,\dots,K$. 
Since
\begin{align*}
\frac{\partial^2 \widetilde{\alpha}_t(\tth,\bb v)}{\partial \theta_j\partial \theta_i}  &= v_2 \frac{\partial^2 \widetilde{\lambda}_{t-1}(\tth)}{\partial \theta_j\partial \theta_i}\left[ \cos(\bb v^\top\widetilde{\tZ}_t(\tth)) -\sin(\bb v^\top\widetilde{\tZ}_t(\tth)) \right]\\
&- v_2^2 \frac{\partial \widetilde{\lambda}_{t-1}(\tth)}{\partial \theta_i}\frac{\partial \widetilde{\lambda}_{t-1}(\tth)}{\partial \theta_j} \widetilde{\alpha}_t(\tth,\bb v),
\end{align*}
we can rewrite 
\begin{align}
\xi_{t,ij}(u,\bb v, \tth)&= T_1(u,\bb v, \tth)\frac{\partial \widetilde{\lambda}_t(\tth)}{\partial \theta_j}\frac{\partial \widetilde{\lambda}_t(\tth)}{\partial \theta_i}+T_2(u,\bb v, \tth)\frac{\partial^2 \widetilde{\lambda}_t(\tth)}{\partial \theta_j\partial \theta_i}\label{eq:rt}\\
&+ T_3(u,\bb v, \tth)\frac{\partial \widetilde{\lambda}_t(\tth)}{\partial \theta_j}\frac{\partial \widetilde{\lambda}_{t-1}(\tth)}{\partial \theta_i} +T_4(u,\bb v, \tth)\frac{\partial \widetilde{\lambda}_t(\tth)}{\partial \theta_i}\frac{\partial \widetilde{\lambda}_{t-1}(\tth)}{\partial \theta_j}\notag\\
&+ T_5(u,\bb v, \tth)\frac{\partial^2 \widetilde{\lambda}_{t-1}(\tth)}{\partial \theta_j\partial \theta_i}+T_6(u,\bb v, \tth)\frac{\partial \widetilde{\lambda}_{t-1}(\tth)}{\partial \theta_j}\frac{\partial \widetilde{\lambda}_{t-1}(\tth)}{\partial \theta_i},\notag
\end{align}
where   $|T_i(u,\bb v,\tth)|\leq 2$ for $i=1,2$, $|T_i(u,\bb v,\tth)|\leq 2 |v_2|$, $i=3,4$, $|T_5(u,\bb v,\tth)|\leq 4|v_2|$ and $|T_6(u,\bb v,\tth)|\leq 2 |v_2|^2$
for all $\tth \in \bb \Theta$, $u\in[0,1]$, $\bb v\in\R^d$. %We show how to treat terms next to $T_1$ and $T_2$, the remaining are analogous. See that
Treating the ``coefficient" of $T_1(\cdot,\cdot,\cdot)$, we see that
\begin{align*}
\frac{\partial \widetilde{\lambda}_t(\tth^*)}{\partial \theta_j}\frac{\partial \widetilde{\lambda}_t(\tth^*)}{\partial \theta_i }&=
 \frac{\partial {\lambda}_t(\tth^*)}{\partial \theta_j}\frac{\partial {\lambda}_t(\tth^*)}{\partial \theta_i } + \frac{\partial \widetilde{\lambda}_t(\tth^*)}{\partial \theta_j}\left[ \frac{\partial \widetilde{\lambda}_t(\tth^*)}{\partial \theta_i } -\frac{\partial {\lambda}_t(\tth^*)}{\partial \theta_i } \right] \\
& + \frac{\partial {\lambda}_t(\tth^*)}{\partial \theta_i }  \left[ \frac{\partial \widetilde{\lambda}_t(\tth^*)}{\partial \theta_j } -\frac{\partial {\lambda}_t(\tth^*)}{\partial \theta_j} \right],
\end{align*}
so for $T$ large enough
\begin{align*}
\left|\frac{1}{T} \sum_{t=2}^T T_1(u,\bb v, \tth^*)\frac{\partial \widetilde{\lambda}_t(\tth^*)}{\partial \theta_j}\frac{\partial \widetilde{\lambda}_t(\tth^*)}{\partial \theta_i}\right|&\leq 
\frac{1}{T}  \sum_{t=2}^T  \left|\frac{\partial {\lambda}_t(\tth^*)}{\partial \theta_j}\frac{\partial {\lambda}_t(\tth^*)}{\partial \theta_i}\right| + \frac{1}{T} \sum_{t=2}^T \sup_{\tth\in \mathcal{V}(\tth_0)}  \frac{\partial \widetilde{\lambda}_t(\tth)}{\partial \theta_j} W_t \rho_2^t\\
& +
  \frac{1}{T}\sum_{t=2}^T  \sup_{\tth\in \mathcal{V}(\tth_0)}  \frac{\partial {\lambda}_t(\tth)}{\partial \theta_i} W_t \rho_2^t, 
\end{align*}
due to Lemma~\ref{lem:1}, hence the two last terms on the right--hand side converge to $0$ almost surely due to Lemma~\ref{lem:limit} and Lemmas~\ref{lem:1} and \ref{lem:3}. Furthermore,
\begin{align*}
\frac{1}{T}  \sum_{t=2}^T  \left|\frac{\partial {\lambda}_t(\tth^*)}{\partial \theta_j}\frac{\partial {\lambda}_t(\tth^*)}{\partial \theta_i}\right| &\leq 
 \frac{1}{T}  \sum_{t=2}^T  \left|\frac{\partial {\lambda}_t(\tth_0)}{\partial \theta_j}\frac{\partial {\lambda}_t(\tth_0)}{\partial \theta_i}\right| \\
 &+\frac{1}{T}  \sum_{t=2}^T  \left|\frac{\partial {\lambda}_t(\tth^*)}{\partial \theta_j}\frac{\partial {\lambda}_t(\tth^*)}{\partial \theta_i}-\frac{\partial {\lambda}_t(\tth_0)}{\partial \theta_j}\frac{\partial {\lambda}_t(\tth_0)}{\partial \theta_i}\right|,
\end{align*}
where as $T\to\infty$, the first term converges  to a finite number due to \eqref{as:Elam1} in Assumption \ref{as:lambda}. Now let $\mathcal{V}_\eta(\tth_0)$ be a neighborhood of $\tth_0$ of radius $1/\eta$. Then for $\eta$ and $T$ large enough $\mathcal{V}_\eta(\tth_0)\subset \mathcal{V}(\tth_0) $
 and the second term can be bounded by
\[
\frac{1}{T}  \sum_{t=2}^T \sup_{\tth \in \mathcal{V}_\eta(\tth_0)} \left|\frac{\partial {\lambda}_t(\tth)}{\partial \theta_j}\frac{\partial {\lambda}_t(\tth)}{\partial \theta_i}-\frac{\partial {\lambda}_t(\tth_0)}{\partial \theta_j}\frac{\partial {\lambda}_t(\tth_0)}{\partial \theta_i}\right|,
\]
which is a sum of stationary and ergodic variables, so as $T \to \infty$ this sum converges to the expectation
\[
\E \sup_{\tth \in \mathcal{V}_\eta(\tth_0)} \left|\frac{\partial {\lambda}_t(\tth)}{\partial \theta_j}\frac{\partial {\lambda}_t(\tth)}{\partial \theta_i}-\frac{\partial {\lambda}_t(\tth_0)}{\partial \theta_j}\frac{\partial {\lambda}_t(\tth_0)}{\partial \theta_i}\right|
\]
which converges to $0$ as $\eta\to\infty$. The ``coefficients" of $T_i(\cdot,\cdot,\cdot)$ , $i=3,4,6$ in \eqref{eq:rt} can be treated in a completely analogous manner, while for the ``coefficients" of $T_i(\cdot,\cdot,\cdot)$ $i=2,5$, in \eqref{eq:rt}, we use Lemma~\ref{lem:4}, otherwise the approximations are the same.

To sum up, we have that
\begin{align*}
\widetilde{B}_T(u,\bb v,\widehat{\tth}) &= B_T(u,\bb v,\tth_0)+\frac{1}{\sqrt{T}} \sum_{t=2}^T \bb s_t(\tth_0,Y_t,\mathcal{I}_{t-1})^\top \bb\beta(u,\bb v) +R_T(u,\bb v),\\
& = Q_T(u,\bb v) + R_T(u,\bb v),
\end{align*}
where
\[
Q_T(u,\bb v) = \frac{1}{\sqrt{T}} \sum_{t=2}^T  \left\{ \left[u^{Y_t} - g(u,{\lambda}_t(\tth_0))\right]\alpha_t(\tth_0,\bb v) +  \bb s_t(\tth_0,Y_t, \mathcal{I}_{t-1})^\top \bb\beta(u,\bb v)\right\}
\]
and 
$
|R_t(u,\bb v)|\leq \frac{C |v_2|^2}{\sqrt{T}}\Pto 0  
$
as $T\to \infty$, for some constant $C<\infty$. Then
\begin{align*}
\left|\widetilde{B}^2_T(u,\bb v,\widehat{\tth}) - Q_T^2(u,\bb v) \right| &\leq
 |R_T(u,\bb v)|^2  + 2|Q_T(u,\bb v)| \cdot |R_T(u,\bb v)| \\
 &\leq \frac{C^2}{T} |v_2|^4 + \frac{2C}{\sqrt{T}} |v_2|^2  |Q_T(u,\bb v)|,
 \end{align*}
 so it follows from Assumption \ref{as:W} that,
 \[
 \int_{0}^1\int_{\R^d} |\widetilde{B}_T(u,\bb v,\widehat{\tth})|^2 w(u) \omega(\bb v) \mathrm{d}u\mathrm{d}\bb v = \int_{0}^1\int_{\R^d} |Q_T(u,\bb v) |^2 w(u) \omega(\bb v) \mathrm{d}u\mathrm{d}\bb v+o_{\mathsf P}(1)
 \] 
 provided that $\int_{0}^1\int_{\R^d} |Q_T(u,\bb v) |^2 w(u) \omega(\bb v) \mathrm{d}u\mathrm{d}\bb v$ converges to a finite distribution.

In order to show this, let $\mathcal{Z}$ be the Gaussian process from the statement of the Theorem~\ref{th1}. It is clear that all the marginal distributions of $\mathcal{Q}=\{Q(u,\bb v), u\in[0,1],\bb v \in\R^d\}$ converge to the marginal distributions of $\mathcal{Z}$. 
We  will show that 
 \begin{equation}\label{eq:convF}
 \int_{0}^1 \int_{F}\bigl[Q_T(u,\bb v)\bigr]^2 w(u)\omega(\bb v) \Dto  \int_{0}^1 \int_{F}  \bigl[\mathcal{Z}(u,\bb v)\bigr]^2 w(u)\omega(\bb v)
 \end{equation}
  for any compact $F\subset \R^d$. To this end, and in view of Theorem~22 from \cite{Ibragimov}  it suffices to verify that 
  \begin{equation}\label{eq:cond1}
  \sup_T \E \int_{0}^1 \int_F Q_T^2(u,\bb v) W(u,\bb v)<\infty
  \end{equation}
  and that there exist constants $C\in(0,\infty)$ and $\kappa>0$ such that
  \begin{equation}\label{eq:cond2}
  \sup_T\E\left|Q_T^2(u_1,\bb v_1) - Q_T^2(u_2,\bb v_2) \right|\leq C \| (u_1,\bb v_1^\top)^\top - (u_2,\bb v_2^\top)^\top \|^{\kappa}.
  \end{equation}

\noindent{\it Condition \eqref{eq:cond1}:} Denote
\[
 \delta_t(u,\bb v)=\left[u^{Y_t} - g(u,{\lambda}_t(\tth_0))\right]\alpha_t(\tth_0,\bb v) +  \bb s_t(\tth_0,Y_t,\mathcal{I}_{t-1})^\top \bb\beta(u,\bb v).
 \]
Then $\{ \delta_t(u,\bb v)\}_{t\in\mathbb{Z}}$ is a sequence of stationary martingale differences with respect to the filtration $\{\mathcal{I}_t\}$, 
$Q_T(u,\bb v) = T^{-1/2} \sum_{t=2}^T  \delta_t(u,\bb v)$,  and
\[
\E Q_T^2(u,\bb v) = \mathsf{Var} Q_T(u,\bb v)= \frac{1}{T} \sum_{t=2}^T  \E[\delta_t(u,\bb v)]^2 =  \E[\delta_1(u,\bb v)]^2
\]
due to the property of martingale differences.
For any $u\in[0,1]$ and $\bb v\in\R^d$, it holds that $|u^{Y_t}|\leq 1$, $|g(u,\lambda(\tth_0))|\leq 1$, $|\alpha_t(\tth_0,\bb v)|\leq 2$, and it follows from \ref{as:g} and \ref{as:lambda} that $\| \bb \beta(u,\bb v)\|\leq H_1$ for some constant $H_1$. So together
\begin{equation*} %\label{eq:deltaT}
|\delta_t(u,\bb v)|\leq 4+\|\bb s_t(\tth_0,Y_t,\mathcal{I}_{t-1})\| H_1,
\end{equation*}
and it follows from the properties of  $s_t(\tth_0,Y_t,\mathcal{I}_{t-1})$ stated in Assumption~\ref{as:est} that
\[
\E[\delta_t(u,\bb v)]^2 \leq  2\E \bigl(16+H_1^2 \|\bb s_t(\tth_0,Y_t,\mathcal{I}_{t-1})\|^2\bigr)<H_2 
\]
for some finite constant $H_2$.
The condition \eqref{eq:cond1} then follows from the properties of the weight function $W$. % in \ref{as:W}.}

\noindent{\it Condition \eqref{eq:cond2}:} The stationarity of $\{ \delta_t(u,\bb v)\}$, and the the property of martingale differences imply that
\[
\E \bigl[Q_T(u_1,\bb v_1) -Q_T(u_2,\bb v_2)\bigr]^2 = \E\left[\delta_1(u_1,\bb v_1)-\delta_1(u_2,\bb v_2)\right]^2.  %\E\left[ \left(u_1^{Y_1} - e^{\lambda_1 (u_1-1) } \right) t(\tZ_1,\bb v_1) - \left(u_2^{Y_1} - e^{\lambda_1 (u_2-1) } \right) t(\tZ_1,\bb v_2)\right]^2.
 \]
 In the following let $\bb s_t:=\bb s_t(\tth_0,Y_t,\mathcal{I}_{t-1})$.
 It follows from the definition of $\bb \beta(u,\bb v)$ that  
 \begin{align*}
 \frac{\partial \delta_t(u,\bb v)}{\partial u} &= Y_t u^{Y_t-1}\alpha_t(\tth_0,\bb v) - \bb s_t^\top \E\left\{ \frac{\partial^2 g(u,z)}{\partial u \partial z}\big\vert_{z=\lambda_1(\tth_0)} \frac{\partial \lambda_1(\tth_0)}{\partial \tth} \alpha(\bb Z_1,\bb v_1)\right\},\\
 \frac{\partial \delta_t(u,\bb v)}{\partial \bb v}  &= \left[u^{Y_t} - g(u,{\lambda}_t(\tth_0))\right]\frac{\partial \alpha_t(\tth_0,\bb v)}{\partial \bb v} \\
 &- \bb s_t^\top \E \left\{\frac{\partial g(u,z)}{\partial z}\vert_{z=\lambda_1(\tth_0)} \frac{\partial \lambda_1(\tth_0)}{\partial \tth} \frac{\partial \alpha(\bb Z_1,\bb v_1)}{\partial \bb v}\right\},
 \end{align*}
 so for all $u\in[0,1]$ and $\bb v\in\R^d$.
 Observe  that
\[
\left|\bb s_t^\top \E\left\{ \frac{\partial^2 g(u,z)}{\partial u \partial z}\big\vert_{z=\lambda_1(\tth_0)} \frac{\partial \lambda_1(\tth_0)}{\partial \tth} \alpha(\bb Z_1,\bb v_1)\right\}\right|\leq 2\|\bb s_t\| \E \left\{ \bigl|H(\lambda_1(\tth_0))\bigr| \left\|\frac{\partial \lambda_1(\tth_0)}{\partial \tth} \right\|\right\}
\]
and
\[
\E \left\{ \bigl|H(\lambda_1(\tth_0))\bigr| \left\|\frac{\partial \lambda_1(\tth_0)}{\partial \tth} \right\|\right\} \leq \sqrt{\E H^2(\lambda_1(\tth_0)) } \sqrt{\E \left\|\frac{\partial \lambda_1(\tth_0)}{\partial \tth} \right\|^2} \leq H_3
\]
for some finite constant $H_3$ due to the Cauchy--Schwartz inequality and Assumption \ref{as:lambda}. Thus,
\[
\left| \frac{\partial \delta_t(u,\bb v)}{\partial u} \right|\leq 2 |Y_t| + 2 H_3 \|\bb s_t\|. 
%\E \left\{ |H(\lambda_1(\tth_0))| \left\|\frac{\partial \lambda_1(\tth_0)}{\partial \tth} \right\|\right\}, 
\]
%where the  expectation is finite due to the Cauchy--Schwartz inequality and Assumption \ref{as:lambda}. 
Similarly,
 \[
  \left\|\ \frac{\partial \delta_t(u,\bb v)}{\partial \bb v} \right\|\leq 4\|\bb Z_t\| + 2\|\bb s_t\| M_1 \E\left\{ \left\|\frac{\partial \lambda_1(\tth_0)}{\partial \tth} \right\| \|\bb Z_1\|\right\},
 \]
  where $M_1$ is from \ref{as:g},  and 
  \[
  \E\left\{ \left\|\frac{\partial \lambda_1(\tth_0)}{\partial \tth} \right\| \|\bb Z_1\|\right\}<H_4
  \]
  for some finite constant $H_4$ due to the Cauchy--Schwartz inequality and assumptions \ref{as:lambda} and \ref{as:stat}. Hence,  the mean value theorem implies that
 \[
\E \bigl[Q_T(u_1,\bb v_1) -Q_T(u_2,\bb v_2)\bigr]^2\leq C_1 \|(u_1,\bb v_1^\top)^\top -(u_2,\bb v_2^\top)^\top  \|^2.
\]
To prove \eqref{eq:cond2} we use again the Cauchy--Schwartz inequality to see that
\begin{align*}
\E |Q_T^2(u_1,\bb v_1) -Q_T^2(u_2,\bb v_2)| &\leq \sqrt{\E |Q_T(u_1,\bb v_1) -Q_T(u_2,\bb v_2)|^2} \sqrt{\E |Q_T(u_1,\bb v_1) +Q_T(u_2,\bb v_2)|^2}\\
&\leq C_2  \|(u_1,\bb v_1^\top)^\top -(u_2,\bb v_2^\top)^\top  \|.
\end{align*}
 for some finite constant $C_2$.

%and hence \eqref{eq:convF} holds. %The rest of the proof proceed similarly as in \cite{hlavka}. 
As the conditions \eqref{eq:cond1}--\eqref{eq:cond2} are satisfied, the convergence in \eqref{eq:convF} holds for any compact $F\subset \R^d$. %follows from \eqref{eq:deltaT} and \ref{as:est} 
We showed 
that $\E Q_T^2(u,\bb v) < H_2$, for a finite constant $H_2$  for all $u\in[0,1]$ and $\bb v \in\R^d$. Also recall that the function $\omega(\bb v)$ is integrable on $\R^{d}$, and therefore for any $\eps>0$ there exists a compact set $F_{\eps}$ such that 
\[
\E \int_0^1 \int_{\R^d\setminus F_{\eps}} |Q_T(u,\bb v)|^2 w(u) \omega(\bb v) \mathrm{d} u \mathrm{d} \bb v
\leq H_2 \int_0^1 w(u)  \mathrm{d} u 
 \int_{\R^d\setminus F_{\eps}}\omega(\bb v)\mathrm{d} \bb v
<\eps.
\]
Analogous arguments apply also to the process 
$\{\mathcal{Z}(u,\bb v), u\ in[0,1], \bb v \in \R^d\}$, so  also $\E \int_0^1 \int_{\R^d\setminus F_{\eps}} |\mathcal{Z}(u,\bb v)|^2 w(u) \omega(\bb v) \mathrm{d} u \mathrm{d} \bb v <\eps$. This finishes the proof. \hfill \qed

%\hfill \qed 

\medskip

\noindent{\it Proof of Theorem~\ref{th1-AR}}

We  sketch the proof solely for $p=1$ while  the situation for $p>1$ can be treated in a completely analogous way.  In the current setup $\lambda_t(\tth)=r(Y_{t-1};\tth)+\pi(\bb X_{t-1};\tth)$, so there is no need to define $\widetilde{\lambda}_t(\tth)$. Moreover,  $\bb Z_t=(Y_{t-1},\bb X_{t-1}^\top)^\top$ does not depend on $\tth$, so using the notation from the proof of Theorem~\ref{th1} we have
\[
\Delta_{T,W} = \int_{0}^1 \int_{\R^d} |B(u,\bb v, \widehat{\tth})|^2 w(u)\omega(\bb v )\mathrm{d}\bb v \mathrm{d}u +o_{\mathsf P}(1).
\] 
Hence, the Taylor expansion is used directly on $B(u,\bb v, \tth)$ and a careful inspection of  the proof of Theorem~\ref{th1} reveals under the assumptions of  Theorem~\ref{th1-AR} it holds that
\[
B(u,\bb v, \widehat{\tth}) = Q(u,\bb v)+R_T(u,\bb v), 
\]
where $R_T(u,\bb v) = o_P(1)$ uniformly in $u\in[0,1]$ and $\bb v \in \R^d$. The rest of the proof then proceeds along the same lines. 
\hfill \qed

\medskip

\noindent{\it Proof of Theorem~\ref{th3}.}
Recall that we use $g(u,\lambda)$ for the PGF $g_{\lambda}(u)$ of $\mathcal{F}_{\lambda}$ and similarly, let $g^A(u,\lambda) = g^A_{\lambda}(u)$ stand for the PGF of $\mathcal{F}^A_{\lambda}$.
It is shown in the proof of Theorem~\ref{th1} that 
\begin{equation}\label{eq:DeltaALT1}
\frac{1}{T}\Delta_{T,W} = \int_{0}^1 \int_{\R^d} |\widetilde{C}(u,\bb v, \widehat{\tth})|^2 w(u)\omega(\bb v )\mathrm{d}\bb v \mathrm{d}u +o_{\mathsf P}(1)
\end{equation}
where $\widetilde{C}(u,\bb v,{\tth}) = T^{-1/2} \widetilde{B}(u,\bb v, \tth)$ for $\widetilde{B}$ given in \eqref{delta-p2}. The Taylor expansion of $\widetilde{C}(u,\bb v, \widehat{\tth})$ in $\tth_0$ gives
\[
\widetilde{C}(u,\bb v, \widehat{\tth}) = \widetilde{C}(u,\bb v, {\tth}_0)  + \frac{\partial \widetilde{C}(u,\bb v,{\tth})}{\partial \tth}\vert_{\tth=\tth^*} (\widehat{\tth}-\tth_0),
\]
where $\tth^*$ is a point between $\tth_0$ and $\widehat{\tth}$. A careful inspection of the proof of Theorem~\ref{th1}  reveals that under the required assumptions 
\[
\frac{1}{T}\Delta_{T,W} = \int_0^1 \int_{\R^d} \Bigl|C_T(u,\bb v,\tth_0) \Bigr|^2 w(u)\omega(\bb v) \mathrm{d} \bb v \mathrm{d} u + o_{\mathsf P}(1) 
\]
for  ${C}_T(u,\bb v,\tth_0)=\frac{1}{T} \sum_{t=2}^T [u^{Y_t}-g(u,\lambda_t(\tth_0))] \alpha(\bb v,\bb Z_t(\tth_0))$.
Hence due to stationarity and ergodicity we have 
\begin{align*}
% \frac{1}{T}\sum_{t=2}^T [u^{Y_t}-g(u,\lambda_t(\tth_0))] \alpha(\bb v,\bb Z_t(\tth_0)) 
C_T(u,\bb v, \tth_0)\Pto &\E [u^{Y_1} \alpha(\bb v,\bb Z_1(\tth_0)) ]-\E[g(u,\lambda_1(\tth_0)) \alpha(\bb v,\bb Z_1(\tth_0))]\\
 =&
  \E [\E u^{Y_1} \alpha(\bb v,\bb Z_1(\tth_0)) | \mathcal{I}_{0}] - \E[g(u,\lambda_1(\tth_0)) \alpha(\bb v,\bb Z_1(\tth_0))] \\
  =& \E\left\{g^A(u,\lambda_1(\tth_0)) - g(u,\lambda_1(\tth_0)) ]\alpha(\bb v,\bb Z_1(\tth_0))\right\}.
\end{align*}
\hfill \qed

\noindent{\it Proof of Theorem \ref{th4}.} 
Observe that in the considered situation $\widehat{\bb Z}_{t,0}=\bb Z_{t,0}$ and  $\{{\bb Z}_{t} \}_{t=p_0+1}^{\infty}$ is a stationary and ergodic sequence. Define 
\[
\lambda_{t,0}(\tth)= r_0(\bb Y_{t-1:t-p_0}; \tth) + \pi_0(\bb X_{t-1};\tth), \quad t\geq p_0+1.
\]
Then $\{\lambda_{t,0}(\tth)\}_{t=p_0+1}^{\infty}$ is also stationary for any $\tth$ and ergodic for $\tth\in\mathcal{V}(\tth_0)$. 
%\wtlambda_{t,A}(\tth)&= h_A(\bb Y_{t-1:t-p_A},\widetilde{\bb \Lambda}_{t-1:t-q_A,A}(\tth),\bb X_{t-1}; \tth), \quad t\geq p_A+1.
%\end{align*}
Since $\lambda_{t,0}$ does not depend on its lagged values,  \eqref{eq:DeltaALT1} holds with $\widetilde{C}_T(u,\bb v,\widehat{\tth})$ replaced by ${C}_T(u,\bb v,\widehat{\tth})$ for
%The test statistic $\Delta_{T,W}$ is of the form
%\[
%\frac{1}{T}\Delta_{T,W} = \int_{0}^1 \int_{\R^d} |{C}(u,\bb v, \widehat{\tth})|^2 w(u)\omega(\bb v )\mathrm{d}\bb v \mathrm{d}u +o_P(1)
%\] 
%for
\begin{equation*} % \label{eq:B-AR}
{C}_T(u,\bb v,\tth)=\frac{1}{T}\sum_{t=2}^T \left[u^{Y_t} - g(u,\lambda_{t,0}(\tth))\right]\alpha\bigl({\bb Z}_{t,0},\bb v\bigr).
\end{equation*}
The Taylor expansion %of $C(u,\bb v, \widehat{\tth})$ in $\tth_0$ 
gives
\[
C_T(u,\bb v, \widehat{\tth}) = C_T(u,\bb v, {\tth}_0)  + \frac{\partial C_T(u,\bb v,{\tth})}{\partial \tth}\vert_{\tth=\tth^*} (\widehat{\tth}-\tth_0),
\]
where $\tth^*$ is a point between $\tth_0$ and $\widehat{\tth}$.
Here $C_T(u,\bb v, {\tth}_0)$ is a sum of stationary and ergodic variables such that
\begin{align*}
C_T(u,\bb v, {\tth}_0)& \Pto \E\left[u^{Y_1} - g(u,\lambda_{1,0}(\tth_0))\right]\alpha\bigl({\bb Z}_{1,0},\bb v\bigr) \\
&=  \E\left\{\left[ g(u,\lambda_{1,A})-g(u,\lambda_{1,0}(\tth_0))\right]\alpha({\bb Z}_{1,0},\bb v)\right\}
%\\
%&=  \E\left\{\left[ g(u,\lambda_1)-g(u,h_0(\bb Y_{0:1-p_0},\bb X_{0};\tth_0))\right]\alpha({\bb Z}_{1,0},\bb v)\right\}
\end{align*}
as $T\to\infty$ and the convergence holds uniformly in $u$ and $\bb v$. 
Furthermore, 
\begin{align*}
\frac{\partial C_T(u,\bb v,{\tth})}{\partial \tth} = -\frac{1}{T}\sum_{t=2}^T \frac{\partial g(u,z)}{\partial z} \vert_{z=\lambda_{t,0}(\tth)} \frac{\partial \lambda_{t,0}(\tth)}{\partial \tth} \alpha(\bb Z_{t,,0}\bb v),
\end{align*}
where
\[
\frac{\partial \lambda_{t,0}(\tth)}{\partial \tth} = \frac{\partial  h_0(\bb Y_{t-1:t-p_0},\bb X_{t-1}; \tth)}{\partial \tth}. %+\frac{\pi_0(\bb X_{t-1};\tth)}{\partial \tth}
\]
It  follows from \eqref{eq:as_g_H1}  and from the properties of $h_0$ that
 $\bigl|{\partial C_T(u,\bb v,{\tth})}/{\partial \tth}\vert_{\tth=\tth^*}\bigr| =O_P(1)$ uniformly in $u\in[0,1]$ and $\bb v \in\R^d$.  The assertion then follows from \eqref{eq:conv.th.alt}.
 \hfill \qed

\bigskip

  For the proof of Theorem~\ref{th5} let $h_0$ be the linear function from \eqref{eq:h0_lin} and $\bb \Theta$ the corresponding parametric space. For $\tth \in\bb\Theta$  define recursively
for $t\geq 2$
\begin{equation}\label{eq:lam.tilde}
\wtlambda_{t,0}(\tth)= \omega+\alpha Y_{t-1}+\beta \wtlambda_{t-1,0}(\tth) +\tgamma^\top\tX_{t-1}
\end{equation}
for some initial $\wtlambda_{1,0} =l_1$.  Let
\[
\widetilde{\bb Z}_{t,0} (\tth)= (Y_{t-1}, \wtlambda_{t,0}(\tth),\bb X_{t-1}^\top)^\top.
\]

\begin{lemma}
Let $\{Y_t\}$ and $\{\tX_t\}$ be stationary and ergodic sequences, $l_1\geq 0$, and $\wtlambda_{t,0}(\tth)$ be defined by \eqref{eq:lam.tilde} for $\tth\in\bb\Theta_0$, where   $\bb \Theta_0$ is a compact subset of $\bb \Theta$.
Then
\[
\wtlambda_{t,0}(\tth) = \omega \frac{1-\beta^{t-1}}{1-\beta}+\sum_{j=0}^{t-2}\beta^j [\alpha Y_{t-1-j}+\tgamma^\top\tX_{t-1-j}] + \beta^{t-1} l_1. 
\]
Furthermore, the derivatives of $\widetilde{\lambda}_{t,0}(\tth)$ are
\begin{align*}
\frac{\partial \widetilde{\lambda}_{t,0}(\tth)}{\partial \omega} &=  \frac{1-\beta^{t-1}}{1-\beta},\\
\frac{\partial \widetilde{\lambda}_{t,0}(\tth)}{\partial \alpha} &=  Y_{t-1} +\beta \frac{\partial \widetilde{\lambda}_{t-1,0}(\tth)}{\partial \alpha} = \sum_{j=0}^{t-2} \beta^j Y_{t-1-j},\\
\frac{\partial \widetilde{\lambda}_{t,0}(\tth)}{\partial \tgamma} &=  \tX_{t-1} +\beta \frac{\partial \widetilde{\lambda}_{t-1,0}(\tth)}{\partial \tgamma} = \sum_{j=0}^{t-2} \beta^j \tX_{t-1-j},\\
\frac{\partial \widetilde{\lambda}_{t,0}(\tth)}{\partial \beta} &= \widetilde{\lambda}_{t-1,0}(\tth)+\beta \frac{\partial \widetilde{\lambda}_{t-1,0}(\tth)}{\partial \beta} 
%= \omega \sum_{j=1}^{t-2} j \beta^{j-1} + \sum_{j=1}^{t-2} j\beta^{j-1}[\alpha Y_{t-1} +\tgamma^\top\tX_{t-1}]  + (t-1)\beta^{t-2}l_1\\
%& 
= \sum_{j=0}^{t-2} \beta^j  \widetilde{\lambda}_{t-1-j,0}(\tth),
\end{align*}
and it holds that
\begin{equation}\label{eq:lam.T.1}
\sup_{\tth\in\bb\Theta_0} \frac{1}{T} \sum_{t=2}^T   \widetilde{\lambda}_{t,0}(\tth)= O_{\mathsf P}(1). 
\end{equation}
and
\begin{equation}\label{eq:lam.T.2}
\sup_{\tth\in\bb\Theta_0}\frac{1}{T} \sum_{t=2}^T \frac{\partial \widetilde{\lambda}_{t,0}(\tth)}{\partial \tth}= O_{\mathsf P}(1). 
 \end{equation}
\end{lemma}

\begin{proof}
The expressions for $\wtlambda_{t,0}(\tth)$ and its derivatives follow from direct computations. 
To prove \eqref{eq:lam.T.1} observe that  
\[
 \frac{1}{T} \sum_{t=2}^T   \widetilde{\lambda}_{t,0}(\tth) = \frac{\omega}{1-\beta} -  \frac{\omega}{1-\beta}  \frac{1}{T}\sum_{t=2}^T \beta^{t-1} + \frac{1}{T}\sum_{t=2}^T \sum_{j=0}^{t-2} \beta^j [\alpha Y_{t-1-j}+\tgamma^\top\tX_{t-1-j}] 
 + l_1 \frac{1}{T}\sum_{t=2}^T \beta^{t-1}. 
\]
Clearly then,
\[
 \frac{\omega}{1-\beta}  \frac{1}{T}\sum_{t=2}^T \beta^{t-1}  \to 0 \quad \text{ and } \quad l_1 \frac{1}{T}\sum_{t=2}^T \beta^{t-1} \to 0
\]
as $T\to\infty$. Furthermore, notice that
\[
 \frac{1}{T}\sum_{t=2}^T \sum_{j=0}^{t-2} \beta^j \alpha Y_{t-1-j} = \alpha \frac{1}{T}\sum_{k=1}^{T-1} Y_k \sum_{j=0}^{T-k-1} \beta^j = \alpha \frac{1}{T}\sum_{k=1}^{T-1} Y_k  \frac{1-\beta^{T-k}}{1-\beta} \leq \frac{\alpha}{1-\beta}\frac{1}{T}\sum_{k=1}^T Y_k. 
\]
The term with $\tgamma^\top\tX_{t-1-j}$ can be treated analogously. 
Since $\bb \Theta_0$ is compact and 
$\{Y_t\}$ and $\{\tX_t\}$ are ergodic, we get \eqref{eq:lam.T.1}. The previous also implies that  \eqref{eq:lam.T.2}  holds for the derivatives with respect to $\omega$, $\alpha$ and $\tgamma$ respectively, so 
  it suffices to show the claim for the derivative with respect to $\beta$. To this end observe that, 
  \[
  \frac{1}{T}\sum_{t=2}^T \frac{\partial \widetilde{\lambda}_{t,0}(\tth)}{\partial \beta}  =   \frac{1}{T}\sum_{t=2}^T \sum_{j=0}^{t-2} \beta^j \widetilde{\lambda}_{t-1-j,0}(\tth)  \leq \frac{1}{1-\beta} \frac{1}{T} \sum_{k=1}^{T-1} \widetilde{\lambda}_{k,0}(\tth),
  \]
where we used the same arguments as before. Thus,  \eqref{eq:lam.T.2}  follows from  \eqref{eq:lam.T.1}.  

\end{proof}

\begin{lemma}\label{lem:lamK} Let $\{Y_t\}$ and $\{\bb X_t\}$ satisfy \ref{as:stat}.
Let $\bb \Theta_0$ be a compact subset of $\bb \Theta$. For  $\tth \in\bb \Theta_0$ define 
\[
{\lambda}_{t,0}(\tth)= 
\frac{\omega}{1-\beta} +\sum_{j=0}^{\infty} \beta^j [\alpha Y_{t-1-j}+\tgamma^\top\tX_{t-1-j}]. 
\]
Then ${\lambda}_{t,0}(\tth)$ is finite with probability 1 and the sequence $\{{\lambda}_{t,0}(\tth) \} $ is stationary and ergodic. Moreover,
there exists $\rho\in(0,1)$ such that
\[
\sup_{\tth \in \bb \Theta_0}|\widetilde{\lambda}_{t,0}(\tth) - \lambda_{t,0}(\tth)|\leq W_t \rho^t, 
\] 
where $W_t$ is a positive random variable such that $\E W_t^{\kappa}<\infty$ for some $\kappa>0$. 
\end{lemma}

\begin{proof}
The finiteness follows together with the ergodicity and stationarity by classical arguments, see \cite[Proposition 3.1.1]{brockwell}. %,  from the properties of $\{Y_t\}$ and $\{\tX_t\}$. 
The inequality follows from the fact that
\[
\lambda_{t,0}(\tth) = \omega+\alpha Y_{t-1}+\tgamma^\top\tX_{t-1} + \beta \lambda_{t-1,0}(\tth)
\]
and by comparing this with \eqref{eq:lam.tilde} which yields
\[
|{\lambda}_{t,0} (\tth)-\widetilde{\lambda}_{t,0}(\tth) |\leq \beta^{t-1}| \lambda_{1,0}-l_1|,
\]
where $l_1$ is the initial value introduced after \eqref{eq:lam.tilde},
and by the fact that $\bb \Theta_0$ is a compact set. 
%Then
%\[
%\frac{1}{T} \sum_{t=K+3}^T \beta^{t-1}l_1 = \frac{l_1\beta^{K+2} (1-\beta^{T-2-K})}{1-\beta}\to 0,\quad T\to\infty,
%\] 
%\begin{align*}
%\frac{  \beta^{K+1} }{T} \sum_{t=K+3}^T \sum_{l=0}^{t-3-K} \beta^{l} \omega &= \frac{\omega \beta^{K+1}}{T (1-\beta)}\left( T-K-2 -\frac{\beta}{1-\beta}(1-\beta^{T-1-K})\right) \\
%&\leq \frac{\omega\beta^{K+1}}{T(1-\beta)}\left(T-K-2+\frac{\beta}{1-\beta}\right)  
%\end{align*}
%and
%\begin{align*}
%\frac{  \beta^{K+1} }{T} \sum_{t=K+3}^T \sum_{l=0}^{t-3-K} \beta^{l} \alpha Y_{t-2-K-l} = \frac{\alpha  \beta^{K+1} }{T} \sum_{j=1}^{T-2-K} Y_j \sum_{l=0}^{T-2-K-j} \beta^{l}  
%\leq  \frac{\alpha  \beta^{K+1} }{T (1-\beta)} \sum_{j=1}^{T-2-K} Y_j 
%\end{align*}
%and by the same arguments 
%\[
%\frac{  \beta^{K+1} }{T} \sum_{t=K+3}^T \sum_{l=0}^{t-3-K} \beta^{l} \tgamma^\top \tX_{t-2-K-l} \leq  \frac{  \beta^{K+1} }{T (1-\beta)} \sum_{j=1}^{T-2-K} \tgamma^\top\tX_j. 
%\]
%The assertion follows from ergodicity of $\{Y_t\}$ and $\{\bb X_t\}$.
\end{proof}

\noindent{\it Proof of Theorem \ref{th5}.}
The test statistic $\Delta_{T,W}$ is of the form
\[
\frac{1}{T}\Delta_{T,W} = \int_{0}^1 \int_{\R^d} |\widetilde{C}(u,\bb v, \widehat{\tth})|^2 w(u)\omega(\bb v )\mathrm{d}\bb v \mathrm{d}u +o_{\mathsf P}(1)
\] 
for
\[
\widetilde{C}_T(u,\bb v,\tth)=\frac{1}{T}\sum_{t=2}^T \left[u^{Y_t} - g(u,\widetilde{\lambda}_{t,0}(\tth))\right]\alpha\bigl( \widetilde{\bb Z}_{t,0}(\tth),\bb v\bigr).
\]
The Taylor expansion of $\widetilde{C}(u,\bb v, \widehat{\tth})$ in $\tth_0$ gives
\[
\widetilde{C}(u,\bb v, \widehat{\tth}) = \widetilde{C}(u,\bb v, {\tth}_0)  + \frac{\partial \widetilde{C}(u,\bb v,{\tth})}{\partial \tth}\vert_{\tth=\tth^*} (\widehat{\tth}-\tth_0),
\]
where $\tth^*$ is a point between $\tth_0$ and $\widehat{\tth}$. Then
\begin{align*}
\left|\frac{1}{T}  \frac{\partial \widetilde{C}(u,\bb v,{\tth})}{\partial \tth}\vert_{\tth=\tth^*}  \right|\leq 2M_1 \sup_{\tth \in \bb \Theta_0} \frac{1}{T} \sum_{t=2}^T \frac{\partial \widetilde{\lambda}_{t,0}(\tth)}{\partial \tth}+4|v_2|  \sup_{\tth \in \bb \Theta_0} \frac{1}{T} \sum_{t=2}^T \frac{\partial \widetilde{\lambda}_{t-1,0}(\tth)}{\partial \tth} = O_{\mathsf P}(1)
\end{align*}
and thus,   
\[
\widetilde{C}(u,\bb v, \widehat{\tth}) = \widetilde{C}(u,\bb v, {\tth}_0)  +R_{1,T}(u,\bb v),
\]
  where $|R_{1,T}(u,v)| \leq \|\bb v\| o_P(1)$. 
 % Let $\lambda_{t,A} = \E[Y_t|\mathcal{I}_{t-1}]$ be the true conditional mean. 
Furthermore,
$
 \widetilde{C}(u,\bb v, {\tth}_0) = A_1(u,\bb v) + A_2(u,\bb v),
$
where
\begin{align*}
A_1(u,\bb v)&= \frac{1}{T}\sum_{t=2}^T \left[u^{Y_t} - g(u,\lambda_{t,A})\right]\alpha\bigl( \widetilde{\bb Z}_{t,0}(\tth_0),\bb v\bigr),\\
A_2(u,\bb v)&= \frac{1}{T} \sum_{t=2}^T \left[g(u,\lambda_{t,A}) - g(u,\lambda_{t,0}(\tth_0))\right]\alpha\bigl( \widetilde{\bb Z}_{t,0}(\tth_0),\bb v\bigr).
\end{align*}
Then $T\cdot A_1(u,\bb v)$ is a sum of martingale differences, which are uniformly bounded for $u\in[0,1]$ and $\bb v\in\R^d$, so they have a finite uniformly bounded variance. Therefore, it follows from the property of martingale differences that  $\E A_1(u,\bb v)=0$ and $\mathsf{Var} A_1(u,\bb v) \to 0$, which shows that $A_1(u,\bb v)\Pto 0$ uniformly in $(u,\bb v^\top)^\top \in[0,1]\times \R^d$. 
%\[
%\var A_1 = \frac{1}{T^2}\sum_{t=1}^T \E\left[\left(u^{Y_t} - g(u;\lambda_{t,A})\right)\alpha\bigl( \widetilde{\bb Z}_{t,0}(\tth_0) \right]^2 \to 0. 
%\]  
On the other hand write,
\[
A_2(u,\bb v) =Q_{T}(u,\bb v) + R_{2,T}(u,\bb v),
\]
where
\[
Q_{T}(u,\bb v) =  \frac{1}{T} \sum_{t=2}^T \left[g(u,\lambda_{t,A}) - g(u,\lambda_{t,0}(\tth_0))\right]\alpha\bigl({\bb Z}_{t,0}(\tth_0),\bb v\bigr) 
\]
and
\begin{align*}
R_{2,T}(u,\bb v)&= \frac{1}{T} \sum_{t=2}^T \left[u^{Y_t} - g(u,\lambda_{t,0}(\tth_0))\right]\left[\alpha\bigl( \widetilde{\bb Z}_{t,0}(\tth_0),\bb v\bigr) - \alpha\bigl({\bb Z}_{t,0}(\tth_0),\bb v\bigr)\right]\\
&+\frac{1}{T} \sum_{t=2}^T \left[ g(u,\lambda_{t,0}(\tth_0)) - g(u,\widetilde{\lambda}_{t,0})(\tth_0))\right] \alpha\bigl( \widetilde{\bb Z}_{t,0}(\tth_0),\bb v\bigr), 
\end{align*}
where we used the identity 
\[
\sum_{i}(a_i-\widetilde{b}_i)\widetilde{c}_i = \sum_i(a_i-b_i)c_i + \sum_i (a_i-b_i)(\widetilde{c}_i-c_i) + \sum_i (b_i-\widetilde{b}_i)\widetilde{c}_i.
\]
Now it follows from \eqref{eq:as_g_H1} that
%According to the assumptions 
$g(\cdot,z)$ is $M_1$-Lipschitz in $z$, and the function $\alpha(\cdot,\cdot)$ is a sum of cosine and sine terms of a scalar product, and that
$
\bb v^\top \widetilde{\bb Z}_{t,0} (\tth)-\bb v^\top {\bb Z}_{t,0} (\tth) = v_2 \bigl[\widetilde{\lambda}_{t,0} (\tth)-{\lambda}_{t,0} (\tth) \bigr].
$
Hence,
\[
|R_{2,T}(u,\bb v)|\leq (4\|\bb v\| +2M_1 )\frac{1}{T} \sum_{t=2}^T  \big|\widetilde{\lambda}_{t,0} (\tth)-{\lambda}_{t,0}(\tth) \bigr| 
\]
and Lemma~\ref{lem:lamK} together with Lemma \ref{lem:limit} imply that $|R_{2,T}(u,\bb v)|\leq \|\bb v\| o_P(1)$. 
Furthermore,  $\{\lambda_{t,0}(\tth_0)\}$ and $\{{\bb Z}_{t,0} (\tth_0)\}$ are  sequences of stationary and ergodic random variables and hence
\[
Q_{T}(u,\bb v) \Pto \E \left[g(u,\lambda_{1,A}) - g(u,\lambda_{1,0}(\tth_0))\right]\alpha\bigl( {\bb Z}_{1,0}(\tth_0),\bb v\bigr) = \zeta(u,\bb v),
\]
which concludes the proof of the assertion. \hfill \qed

\end{document}